\documentclass[11pt]{article}

\usepackage[margin=1in]{geometry}
\usepackage{amsmath,amssymb}
\usepackage{graphicx}
\usepackage{cite}
\usepackage{hyperref}
\usepackage{caption}
\usepackage{booktabs}
\usepackage{algorithm}
\usepackage{algpseudocode}
\usepackage{enumitem}
\usepackage{fancyhdr}
\usepackage{ifthen}
\usepackage{etoolbox}
\usepackage{subcaption}

\usepackage{moreverb}

\usepackage{siunitx}
\usepackage{amsmath,amsthm,amssymb,epsf}
\usepackage{bm} 
\usepackage{graphicx,psfrag}
\usepackage{epstopdf}
\usepackage{verbatim}
\usepackage{multirow}
\usepackage{color}
\usepackage{setspace}
\usepackage{enumitem}
\usepackage{algorithm}
\usepackage{algpseudocode}
\usepackage[algo2e]{algorithm2e}
\usepackage{algorithmicx}
\algrenewcommand\algorithmicrequire{\textbf{Inputs:}}
\algrenewcommand\algorithmicensure{\textbf{Output:}}
\algnewcommand\algorithmicsteps{\textbf{Steps:}}
\algnewcommand\Steps{\item[\algorithmicsteps]}
\usepackage{hyperref}

\usepackage{tabularx}
\usepackage{booktabs}

\usepackage{tikz}
\usetikzlibrary{shapes.geometric, arrows.meta, positioning}

\tikzstyle{block} = [rectangle, draw, rounded corners, text centered, minimum height=2em, minimum width=10em, fill=blue!10]
\tikzstyle{arrow} = [thick, ->, >=stealth]

\newtheorem{assumption}{Assumption}
\newtheorem{lemma}{Lemma}
\newtheorem{remark}{Remark}
\newtheorem{theorem}{Theorem}

\newtheorem{definition}{Definition}
\newtheorem{problem}{Problem}
\newtheorem{proposition}{Proposition}

\patchcmd{\maketitle}{\thispagestyle{plain}}{\thispagestyle{empty}}{}{}

\title{\bfseries Risk-Aware Safe Reinforcement Learning for Control of Stochastic Linear Systems}
\author{
    Babak Esmaeili, Nariman Niknejad, and Hamidreza Modares\thanks{Corresponding author} \\
    \\  
    \small Department of Mechanical Engineering, Michigan State University, MI, USA \\
    \small \texttt{\{esmaeil1, niknejad, modaresh\}@msu.edu}
}
\date{}

\begin{document}
\maketitle

\begin{abstract}
This paper presents a risk-aware safe reinforcement learning (RL) control design for stochastic discrete-time linear systems. Rather than using a safety certifier to myopically intervene with the RL controller, a risk-informed safe controller is also learned besides the RL controller, and the RL and safe controllers are combined together. Several advantages come along with this approach: 1) High-confidence safety can be certified without relying on a high-fidelity system model and using limited data available, 2)  Myopic interventions and convergence to an undesired equilibrium can be avoided by deciding on the contribution of two stabilizing controllers, and 3) highly efficient and computationally tractable solutions can be provided by optimizing over a scalar decision variable and linear programming polyhedral sets. 
To learn safe controllers with a large invariant set,  piecewise affine controllers are learned instead of linear controllers. To this end, the closed-loop system is first represented using collected data, a decision variable, and noise. The effect of the decision variable on the variance of the safe violation of the closed-loop system is formalized. The decision variable is then designed such that the probability of safety violation for the learned closed-loop system is minimized.  It is shown that this control-oriented approach reduces the data requirements and can also reduce the variance of safety violations. 
Finally, to integrate the safe and RL controllers, a new data-driven interpolation technique is introduced. This method aims to maintain the RL agent's optimal implementation while ensuring its safety within environments characterized by noise. The study concludes with a simulation example that serves to validate the theoretical results. 
\end{abstract}

\vspace{1em}
\noindent\textbf{Keywords:} Model-Free Control, Probabilistic Control, Reinforcement Learning, Convex Hull

\section{Introduction}
Reinforcement learning (RL) control has recently received a surge of interest due to its pivotal role in enabling autonomous control systems that must operate in dynamic and uncertain environments. In RL, as a branch of machine learning, an agent learns optimal control policies through its interactions with the environment to maximize cumulative rewards. RL has already demonstrated promising capabilities in complex control tasks, such as the control of the non-affine yaw channel of helicopters via off-policy RL methods \cite{zhang2025data} and decision-making for autonomous vehicles using iterative single-critic learning frameworks \cite{zhang2021adaptive}. However, while these successes highlight the potential of RL in real-world applications, most RL methods optimize performance without explicitly considering safety constraints. In safety-critical domains, ensuring that agents act safely during both learning and deployment is vital to prevent undesirable outcomes or catastrophic failures.


To avoid these undesirable outcomes and failures, safe RL holds the promise of enabling autonomous systems to make decisions that are both efficient and safe, opening avenues for applications across diverse domains, from autonomous vehicles and robotics to healthcare and industrial automation. Recent advances in Safe Reinforcement Learning (Safe RL) have sought to address these safety challenges by explicitly incorporating state or action constraints during both learning and deployment phases. The concept of safety in reinforcement learning has been interpreted in various ways across different research directions. One approach defines safe RL as providing risk-aware guarantees, where the likelihood of deviating from a nominal trajectory remains below a specified threshold \cite{zhang2024cvar}. Another common formulation treats safe RL as a constrained Markov decision process (CMDP), aiming to maximize cumulative rewards while keeping the expected cumulative cost under a set limit \cite{wachi2020safe}. However, many real-world scenarios require safety to be enforced continuously, not just on average. As a result, another line of research defines safe RL as optimizing performance while strictly satisfying safety constraints at every time step \cite{konighofer2023online}, typically by ensuring the agent's state remains within a predefined admissible set throughout the learning and deployment phases. In this paper, we formalize a safe RL that guarantees instantaneous satisfaction of safety constraints.

Safety certificates have been extensively employed to provide learning-enabled agents with verifiable safety assurances \cite{brunke2022safe,qin2023adaptive,zanon2020safe,grandia2021multi,mazouchi2021conflict,zhang2021safe,li2020robust}. These safety credentials typically harness control barrier functions (CBFs) to provide myopic fixes to the RL agent actions \cite{gao2022safety,yang2019self,marvi2021barrier,ahmadi2021risk,liu2024fully,zeng2021safety,seo2022safety,chern2021safe,agrawal2017discrete}. This myopic intervention with the RL actions can result in reaching undesired equilibrium points \cite{reis2020control} and yielding poor performance due to frequent interventions. 
Besides, CBF methodologies heavily rely on precise system models. This limitation makes the practical deployment of safe reinforcement learning methods particularly challenging in real-world systems such as autonomous vehicles, aerospace platforms, and robotic manipulators, where obtaining accurate models is difficult and stochastic disturbances are inherent. In such applications, safe control frameworks must not only provide formal guarantees but also operate reliably with noisy and incomplete empirical data. Consequently, there is a strong need for data-driven methods that can directly synthesize safe controllers from available data without requiring full system identification or restrictive modeling assumptions. When a system model is not available, data-driven control methods can be highly advantageous in reducing conservatism and adapting to the situations. Indirect data-driven control (i.e., model-based control) methods learn a system model first and then leverage it to design a control that reaches desired specifications. Direct data-driven control (i.e., model-free) methods bypass learning a system model and directly learn a controller from collected data. 
Nonetheless, indirect learning approaches may not be suited for safety-critical systems primarily for the following reasons. Firstly, they can only develop a system model once specific data conditions relating to state-input data richness are fulfilled. Since data collection is costly and risky in safety-critical settings, relaxing these data prerequisites is pivotal for the efficacy of future autonomous systems. Secondly, the variability of the learned open-loop system is contingent on the signal-to-noise ratio (SNR) of collected data and remains unaffected by control mechanisms. Hence, leveraging control-oriented learning approaches to reduce variability in safety breaches given the available data emerges as a necessity for enhancing safety. Lastly, model-based CBF techniques for stochastic systems are confined to scenarios where noise has a finite range \cite{samuelson2018safety,ahmadi2021risk}.

Direct data-driven methods have gained a surge of interest to devise safe or optimal control strategies \cite{bisoffi2020data,luppi2022data,bisoffi2022controller,de2021low,modares2023}. However, the current scope of research into direct data-driven safe control is restricted to deterministic systems or involves treating noise as either a bounded disturbance, leading to the creation of robust but conservative controllers for the system \cite{}, or as a measurable signal \cite{modares2023}. Unfortunately, the efficacy of robust control diminishes when confronted with systems where noise follows a distribution with infinite support. Additionally, noise is often not practically measurable in real-world scenarios. Notably, in references \cite{bisoffi2022controller,de2021low}, optimal controllers grounded in certainty equivalence principles are formulated for stochastic linear systems. Nevertheless, the analysis of stability and performance is carried out only in hindsight. Consequently, these guarantees pertain solely to the nominal model and predicted outcomes. Disregarding the noise variance in safety violations can engender performance fluctuations when implementing these controllers in practical systems.

Another challenge with direct data-driven safe control is that they mainly leverage set-theoretic control design tools \cite{bisoffi2020data,modares2023safe}. This method typically uses the concept of $\lambda$-contractivity to design controllers that make a given admissible set invariant for the closed-loop system while making the trajectories converge to the origin with a speed of $\lambda$. Set invariance guarantees that starting from inside the set, the system's states will not leave the set in some sense; thus, the set remains safe. However, as the complexity of the system and/or the admissible set increases, it becomes increasingly difficult to make the entire admissible set invariant using set-theoretic tools \cite{blanchini1999set}. In practice, admissible sets are sets for which the system's states are allowed to evolve inside of them and are often defined by the physical limitations of the system and its environment. As a result, designing controllers that can make any desired admissible set invariant is a daunting task \cite{blanchini2008set}. The invariant set is typically a subset of the admissible set, and its size depends on the data richness and the control structure. 

Partitioning complex polyhedral admissible sets into disjoint polyhedral sets is a promising approach for designing controllers for complex admissible sets that cannot be made entirely safe or invariant using just a linear feedback controller \cite{nguyen2022convex,hoai2023further}. These partitioning-based methods, however, are limited to deterministic systems with known dynamics. For systems under noise and uncertain dynamics, a probabilistic or high-confidence risk-informed safe controller must be designed. Besides, the size of the safe set inside the admissible set depends on the data quality and the risk level the system can tolerate. Therefore, it is of vital importance to design data-based controllers that minimize the risk of safety violations given only the available data. 

Motivated by the practical challenges discussed above, particularly the need for scalable safe learning frameworks that operate effectively under uncertainty and with limited data, we propose a novel approach to safe reinforcement learning that is both risk-aware and data-driven. The goal is to bridge the gap between theoretical safety guarantees and practical deployment requirements in stochastic control systems, where model inaccuracies, noise, and data collection limitations present significant obstacles. In this paper, we first introduce a safe feedback control policy that makes the convex hull of a known number of ellipsoids $\lambda$-contractive in expectation. By imposing a set containment condition to ensure that the convex hull set is inside the admissible set or covers it entirely if possible, safety is guaranteed for a maximum-size set inside the admissible set. It is shown that this approach is risk-neutral since it only guarantees safety in expectation. A risk-informed piecewise-affine safe control design is highly desirable due to its robustness guarantees, especially when the system model is uncertain, increasing the risk of safety violation. Therefore, next, a direct data-driven risk-informed piecewise-affine safe control approach is introduced to minimize safety violation variance and maximize the size of the safe set. To this end, a control-oriented approach is taken in which the closed-loop system model is directly characterized by data and a decision variable, and the control gains of the piecewise-affine controller (which are a function of the decision variable) are learned to certify safety with minimum variance on its violation directly. This control-oriented approach demands less data than existing indirect learning methods while offering reduced safety violation risk. Compared to traditional methods, the proposed framework offers several key advantages. It enables probabilistic safety with risk-awareness by minimizing the variance of safety violations, operates directly from empirical data without requiring explicit system identification, and constructs scalable safe sets over complex admissible spaces using convex hull methods. Additionally, the lightweight scalar optimization between the safe controller and the RL controller minimizes intervention frequency, preserving both safety and task performance in stochastic environments. The learned risk-informed safe controller is then integrated with any RL controller to certify its safety with high probability. Instead of only using a safety certificate to intervene with the RL actions, a learned safe controller is integrated with an RL controller, ensuring both safety and performance guarantees. A novel data-based optimization over a scalar is presented to determine the contribution of each controller at each point in time. The effectiveness of the proposed approach is demonstrated through a simulation example. \vspace{6pt}

\section{NOTATIONS AND PRELIMINARIES}

Throughout the paper, the Kronecker product is denoted by $\otimes$, and the identity matrix of appropriate dimension is represented as $I$. The set of positive semi-definite $n \times n$ matrices is represented by $\mathbb{S}^n$. For a matrix $A$, $A_i$ indicates its $i$-th row, and $A_{ij}$ represents the element in the $i$-th row and $j$-th column of $A$. For matrices or vectors $A$ and $B$ with the same dimensions, $A (\leq, \geq) B$ denotes a component-wise inequality, where ${A_{ij}} (\leq, \geq) {B_{ij}}$ holds for all $i$ and $j$. For a matrix $Q$, $Q (\preceq,\succeq) 0$ implies that $Q$ is negative or positive semi-definite. Given a set $\mathcal{S}$ and a scalar $\mu \ge 0$, $\mu \mathcal{S}$ is defined as the set of all $\mu x$ where $x$ belongs to $\mathcal{S}$. When dealing with symmetric matrices, the symbol ($*$) is used to denote each of the symmetric blocks within the matrix. The frontier of a given set $\mathcal{S}$ is denoted as $\operatorname{Fr}(\mathcal{S})$.

The convex hull formed by the sets $\mathcal{S}_1,\mathcal{S}_2,\ldots,\mathcal{S}_n$ is denoted as $\mathcal{S} = \operatorname{Co}(\mathcal{S}_1,\mathcal{S}_2,\ldots,\mathcal{S}_n)$. Any element of the convex hull, i.e., any $x \in \mathcal{S}$, can be expressed as a weighted combination of elements from the sets $\mathcal{S}_1,\mathcal{S}_2,\ldots,\mathcal{S}_n$. That is, 
\begin{equation}\label{eq.combination}
x=\alpha_1x_1+\alpha_2x_2+\ldots+\alpha_nx_n,
\end{equation}
for some $x_1 \in \mathcal{S}_1$, $x_2 \in \mathcal{S}_2,\ldots,x_n \in \mathcal{S}_n$, along with weights $\alpha_1,\alpha_2,\ldots,\alpha_n$ such that $\sum\limits_{i=1}^n \alpha_i = 1$ and $0 \leq \alpha_i \leq 1$.

\begin{definition}\label{def_0}
For any two positive integers $a$ and $b$, $\operatorname{mod}(a,b)$ denotes the remainder of their division. Given a set of elements with a fixed size $M$, the rotational indexing function $\operatorname{R_m}(i)$ maps an index $i$ to another index $j$ in a circular or cyclic manner. In this paper, the mapping function $\operatorname{R_m}(.)$ is defined as
\begin{equation}\label{eq.mapping}
j=\operatorname{R_m}(i)=\operatorname{mod}(i+M-2,M)+1.
\end{equation}
\end{definition}

Let the random variables be defined on a probability space denoted as $(\Gamma,\mathcal{F},\mathbb{P})$. Here, $\Gamma$ represents the sample space, $\mathcal{F}$ is the associated $\sigma$-algebra, and $\mathbb{P}$ denotes the probability measure. For a random variable $\nu : \Gamma \rightarrow \mathbb{R}^n$ defined on this probability space, the notation $\nu \in \mathbb{R}^n$ indicates its dimension. The mathematical expectation of $\nu$ is denoted as $\mathbb{E}[\nu]$, and if $\mathbb{E}[\nu]=\hat{\nu}$, the covariance of $\nu$ can be computed using the formula $\mathbb{E}[(\nu-\hat{\nu})(\nu-\hat{\nu})^T]$. For a random vector $\nu \in \mathbb{R}^{n \times 1}$, the following lemma holds.

\begin{lemma}\label{lem_1}
\cite{coppens2020datadriven} For a given random vector $\nu \in \mathbb{R}^{n \times 1}$ and a matrix $Q \in \mathbb{R}^{n \times n}$, one has
\begin{equation}\label{eq.Variance_1_2}
\mathbb{E}[\nu^T Q \nu] = \mathrm{Tr}(Q\mathbb{E}[\Tilde{\nu}\Tilde{\nu}^T])+\mathbb{E}[\nu]^T Q \mathbb{E}[\nu],
\end{equation}
where $\Tilde{\nu}=\nu-\mathbb{E}[\nu]$.
\end{lemma}

\vspace{-6pt}

The following definitions are provided to define sets that will be used in this paper to characterize admissible and safe sets.

\begin{definition}\label{def_1} 
\cite{blanchini2008set} A C-set is a set that is both convex and compact, and its interior contains the origin.
\end{definition} 
\vspace{-9pt}

\begin{definition}\label{def_2} 
\cite{blanchini2008set} A polyhedral C-set, denoted by ${\mathcal {S}} (F,g)$, is represented by
\begin{align}
{\mathcal {S}} (F,g) \, & = \{ x \in {\mathbb{R}^n}:Fx \le g\}  \nonumber \\
& = \{ x \in {\mathbb{R}^n}:{F_l}x \le g_l,\, \, \, l = 1,\ldots,q \},
\label{poly}
\end{align}
where $F \in {\mathbb{R}^{q \times n}}$ is a matrix with $q$ rows, i.e., 
${F_l}$ for $l = 1,\ldots,q$, and $g$ is a vector with elements $g_l$, $l = 1,\ldots,q$.
\end{definition}

\begin{definition}\label{def_3} 
\cite{blanchini2008set} For a given positive-definite matrix $P$, an ellipsoidal C-set is denoted by
\begin{align}\label{eq.safe_set_ellipsoid}
\mathcal{E}(P,1) \, & = \{ x \in {\mathbb{R}^n}:x^TP^{-1}x \leq 1\}.
\end{align}
\end{definition} 

\begin{lemma}\label{lem_2}
\cite{geng2019data} Assume that there is a joint chance constraint denoted by 
\begin{equation}\label{eq.joint_chance}
\mathbb{P}[Hx+Mw \leq g] \geq (1-\epsilon),
\end{equation}
where $x \in \mathbb{R}^n$ represents the decision variable, $w$ is a random variable with a normal distribution $\mathcal{N}(0,\Sigma)$, $H$ and $M$ are matrices with dimensions $q \times n$, and $g$ is a vector in $\mathbb{R}^q$. Now, if the constraints 
\begin{equation}\label{eq.joint_chance_j}
H_jx+M_j\mu \leq g_j-k_j\sqrt{M_j\Sigma M_j^T}
\end{equation}
are satisfied for all $j=1,\ldots,q$, where $H_j$ and $M_j$ are the $j$-th rows of matrices $H$ and $M$, respectively, $k_j=\sqrt{\frac{1-\epsilon_j}{\epsilon_j}}$, and $\sum \limits_{j} \epsilon_j \leq \epsilon$, then the original joint chance constraint \eqref{eq.joint_chance} is also satisfied.
\end{lemma}

In Lemma 2, $k_j$ is a constant, and $\epsilon_j$ represents the accepted probability of violation of the constraint $H_jx+M_jw \leq g_j$.

\section{Problem Formulation}
Consider the following discrete-time linear time-invariant (LTI) system
\begin{equation}\label{eq.LTI}
x(t+1) = Ax(t)+Bu(t)+w(t),
\end{equation}
where $A \in \mathbb{R}^{n \times n}$ is the system matrix and $B \in \mathbb{R}^{n \times m}$ denotes the input matrix. Moreover, $x(t) \in \mathbb{R}^n$ and $u(t) \in \mathbb{R}^m$ represent the system states and control input at time-step $t$, respectively, and $w(t)$ is the system noise.

\begin{assumption}\label{Assumption_1}
The vector $w(t)=[w_1(t),\ldots,w_n(t)]^T$ representing the noise in the system \eqref{eq.LTI} is assumed to have a Gaussian distribution. It has a mean of zero and a variance of $\Sigma$, denoted as $w \sim \mathcal{N}(0,\Sigma)$ where $\mathbb{E}[w_i(t)w_j(t)]=0$ for $i \neq j$, and $\mathbb{E}[w_i^2(t)]=\sigma_i^2$ for $i=1,\ldots,n$.
\end{assumption}

\begin{assumption}\label{Assumption_2}
The unknown matrix pair $(A,B)$ is stabilizable.
\end{assumption}

Prior to describing the problem, we emphasize the significance of contractive sets as a primary technique for ensuring safety. To clarify this concept, the following definitions that establish a framework for maintaining the system within a predetermined set of states over time are first provided. This framework is crucial for applications that prioritize safety and assists in designing controllers capable of enforcing set boundaries.

\begin{definition}\label{def_4} 
(\textbf{Contractive Set for Deterministic Systems, i.e., the system \eqref{eq.LTI} with $w(t) \equiv 0$}):
If for every $x(t) \in \mathcal{S} \subseteq \mathbb{R}^n$, it holds that $x(t+1) \in \lambda\mathcal{S}$ for all $t \geq 0$, where $0 < \lambda \leq 1$, then $\mathcal{S}$ is referred to as a $\lambda$-contractive set.
\end{definition}

\begin{definition}\label{def_5} 
\cite{modares2023} (\textbf{Contractive Set in Expectation (CSiE)}): A set $\mathcal{S} \subseteq \mathbb{R}^n$ is called $\lambda$-contractive in expectation for the system \eqref{eq.LTI} if $x(t) \in \mathcal{S}$ implies that $\mathbb{E}[x(t+1)] \in \lambda\mathcal{S}$ $\forall t \geq 0$.
\end{definition}

\begin{definition}\label{def_6} 
\cite{modares2023} (\textbf{Contractive Set in Probability (CSiP)}): A set $\mathcal{S} \subseteq \mathbb{R}^n$ is called $\lambda$-contractive in probability for the system \eqref{eq.LTI} if $x(t) \in \mathcal{S}$ implies that $\mathbb{P}[x(t+1) \in \lambda\mathcal{S}] \geq (1-\epsilon)$ $\forall t \geq 0$, where $\epsilon$ is an acceptable risk level.
\end{definition}

\begin{definition}\label{definition_8}
(\textbf{Admissible set}): An admissible set is defined according to the permissible physical boundaries within which the system is allowed to operate.
\end{definition}


\begin{definition}\label{definition_9}
(\textbf{Safe set}): A subset of an admissible set is called a safe set if it is invariant in some sense. That is, starting from the safe set, the system's trajectories do not leave it in some sense. 
\end{definition}



It is shown in \cite{bisoffi2020data} that for a deterministic system, a $\lambda$-contractive set is an invariant set and thus is a safe set. That is, if a set $\mathcal{S}$ is $\lambda$-contractive, and if $x(0) \in \mathcal{S}$, then $x(t) \in \mathcal{S}, \,\,\, \forall t \ge 0$. For stochastic systems, it is shown in \cite{modares2023} that if the set $\mathcal{S}$ is CSiE (CSiP), then the set is safe in expectation (in probability). That is, if $x(0) \in \mathcal{S}$, then $\mathbb{E}[x(t)] \in \mathcal{S}, \,\,\, \forall t \ge 0$ \big($\mathbb{P}[x(t) \in \mathcal{S}] \geq (1-\epsilon), \,\,\, \forall t \geq 0$ for some risk level $\epsilon$\big). In this paper, safety in probability is considered since it provides more robustness compared to safety in expectation. The former is risk-aware, while the latter is risk-neutral.

While the system trajectories are allowed to evolve within the admissible set, it is not always possible to make the entire admissible set safe. The size of the safe set depends on the data quality and the control structure. Therefore, to improve safety, the controller must maximize the size of the safe set based on the available data and the control structure. Existing linear controllers limit the size of the safe set, which can significantly limit the maneuverability of the RL agent. Therefore, in this paper, data-based piecewise-affine nonlinear controllers are designed for safety. The following problem formalizes the safe optimal control problem.

\begin{problem}\label{Problem_1}
(\textbf{Safe Optimal Control}): Consider the given system \eqref{eq.LTI}. Our objective is to design a control policy $\pi(t)=u\big(x(t)\big)$ by solving the following constrained optimal control problem
\begin{align}\label{eq.Problem}
\begin{array}{l}
\mathop{\arg \min}\limits_{\pi} \,\,\, J\big(x(t),\pi(t)\big) \\
\mathrm{s.t.} \,\,\,\,\,\, \mathbb{P}[x(t) \in \mathcal{S}] \geq (1-\epsilon), \,\,\, \forall t \geq 0, \,\,
\end{array}
\end{align}
where the cost function $J\big(x(t),\pi(t)\big)$ is defined as
\begin{align}\label{eq.cost_RL}
J\big(x(t),\pi(t)\big) = \mathbb{E}\big[\sum_{t=0}^\infty {\gamma^{t} r\big(x(t),\pi(t)\big)}\big],
\end{align}

Here, $\mathcal{S}=\{x:h(x) \geq 0\}$ represents a pre-specified admissible set that includes hard constraints based on the safety function $h(x)$, and $\epsilon$ specifies an acceptable risk level. Additionally, $0 < \gamma \leq 1$ is a positive discount factor, and $r\big(x(t),\pi(t)\big)$ denotes the reward function that implicitly reflects the desired specifications.
\end{problem}


\begin{remark}
Problem~\ref{Problem_1} aligns with the formulation of Safe Reinforcement Learning (Safe RL), where the goal is to maximize performance while ensuring probabilistic safety constraints are satisfied. It requires that the system's state remains within the admissible set \( \mathcal{S} \) with high probability, while also minimizing the expected cumulative cost. However, due to the presence of probabilistic constraints over an infinite horizon, directly solving this problem is computationally intractable in general. This motivates the need for approximate solutions that decouple performance and safety, as discussed below. In our proposed approach, this decoupling is handled via a data-driven risk-aware safe controller that supervises the RL agent with minimal intervention, as later described in Section VIII.
\end{remark} 

\begin{assumption}\label{Assumption_3}
The admissible set is described as a polyhedral set that remains unchanged over time. It is represented as a polyhedral set ${\cal{S}}(F,g)$ defined in \eqref{poly}, for which the safety function $h(x)$ is also defined as $h(x) = g - Fx$.
\end{assumption}

Finding a feedback controller that solves Problem 1 is computationally intractable even for systems with known dynamics and even for the simplest case of using linear controllers for time-invariant C-set constraints defined by the function $h(x)$. Consequently, instead of directly addressing the optimization problem, existing RL algorithms separate safety and performance concerns: They first learn an unconstrained control policy $u^*$ that minimizes the cost function $J$ in \eqref{eq.cost_RL} without considering physical constraints (assuming $\mathcal{S}=R^n$). Subsequently, a model-based safety certifier or shield is utilized to make minimal adjustments to the RL's actions while ensuring safety. For deterministic systems, this implementation involves solving the following optimization problem where the constraints act as a shield, certifying the safety of the RL actions prior to deployment \cite{cheng2019end}.
\begin{align}\label{eq.shield}
& {u^s} = \mathop{\arg \min}\limits_{u} \,\,\, {(u-u^*)^T}(u-u^*) \nonumber \\
& \mathrm{s.t.} \,\,\,\,\,\,\,\,\,\,\, h\big(x(t+1)\big)-h\big(x(t)\big)+\rho h\big(x(t)\big) \geq 0, \,\, \forall t \geq 0, \,\,\, \rho \leq 1,
\end{align} 

\noindent in which the constraint refers to a barrier certification constraint that ensures the set $\mathcal{S}$ remains invariant, and $u^s$ signifies the safe optimal control input applied to the system.

Nevertheless, this approach has certain drawbacks. Firstly, it requires complete knowledge of the system dynamics, which may not always be available. Secondly, acquiring a model of the system using data can be data-intensive, creating a bottleneck in certifying safety under uncertainties. Additionally, in stochastic systems, constructing a robust controller using a worst-case model often leads to excessively conservative behavior, which in turn degrades performance as the safety mechanism tends to intervene frequently alongside the RL controller. Moreover, these interventions are often executed in a myopic manner, correcting RL actions locally without considering long-term task performance. This lack of foresight can inhibit the RL agent’s ability to explore optimally or to converge to high-performing policies. In contrast, our proposed approach addresses this limitation by learning an unconstrained RL policy and a risk-aware safe control policy separately and then merging them together to optimize the performance while ensuring safety. This approach allows learning for the sake of safety in a closed-loop manner, which reduces conservatism. Besides, in sharp contrast to the CBF approaches, which, in general, are non-convex for discrete-time systems, our approach requires only solving an online scalar convex optimization that interpolates the two policies. This allows the system to enforce probabilistic safety constraints in a more global and adaptive manner while minimizing interference with the RL controller’s autonomy. Besides, the proposed approach guarantees the stability of the system if the safe set is compact. Thus, the proposed method achieves a better balance between safety and performance than myopic CBF-based corrections. The stochastic counterpart of the CBF constraint in \eqref{eq.shield} is limited to noises with finite support \cite{samuelson2018safety,ahmadi2021risk}, further restricting its applicability in practical scenarios involving unbounded stochastic disturbances. Lastly, a study conducted in \cite{reis2020control} has shown that imposing both a control Lyapunov stability constraint and a CBF-based safety constraint can cause undesired convergence towards an equilibrium solution, highlighting potential conflicts between safety and performance objectives in such frameworks.

Similar to existing safe RL algorithms, we separate safety and optimality concerns. However, in sharp contrast to previous results, we learn two different control policies (i.e., a safe control policy and an RL control policy) and merge them together rather than learning only an RL control policy and myopically intervening with it. Our approach is RL-agnostic and will certify the safety of any RL algorithm. The safe controller is learned to avoid limiting the maneuverability of the RL controller as much as possible, thus significantly reducing conflict. This is because, in the optimization Problem 1, since the entire admissible set $\mathcal{S}$ cannot be made invariant or safe in general, the safe controller, when merged with the RL controller, will confine the RL system trajectories to a subset of $\mathcal{S}$ which is invariant in probability. That is, the constraint $ \mathbb{P}[x(t) \in \mathcal{S}] \geq (1-\epsilon)$ will actually be satisfied by ensuring $\mathbb{P}[x(t) \in \mathcal{S}_c] \geq (1-\epsilon)$ where $\mathcal{S}_c \subseteq \mathcal{S}$ is the safe set. Therefore, maximizing the size of the safe set $\mathcal{S}_c \subseteq \mathcal{S}$ is crucial to improving RL performance. 

To this end, two different approaches are presented to improve safety. First, a certainty equivalence-based direct learning technique is developed, enabling the acquisition of a risk-neutral safety backup policy. Second, a risk-informed piecewise-affine controller is learned for safety that maximizes the size of the safe set. This is in sharp contrast to existing learning-based safe controllers that are limited to linear controllers with restricted regions of attraction and are typically risk-neutral. This controller not only seeks to maximize the size of the safe set but also takes into account the quality of the available data. In particular, the quality of data impacts not just the estimation accuracy but also the variability of the closed-loop behavior. To explicitly account for this, the proposed controller not only ensures that the expected state lies within the safe set but also minimizes the variance of constraint violations. This is achieved through a variance-aware formulation that optimizes the spread of the state distribution using a data-driven characterization of the closed-loop dynamics. As a result, the controller ensures high-probability satisfaction of safety constraints under stochastic disturbances, rather than merely providing guarantees in expectation. By adapting the safe set's size based on data quality, we aim to find a balance between safety and performance, allowing the RL agent to perform in complex and dynamic environments. The RL and safe control policies are finally merged through linear programming optimization to determine their contributions over time.

\vspace{-6pt}

\section{Open-loop Safety using Piecewise-affine Controllers}
This section presents a solution to certify the largest safe set (i.e., invariant set) of a deterministic linear system inside an admissible set. This approach will then be leveraged to design controllers that maximize the size of the safe set inside an admissible set. To approximate the safe set, the concept of the convex hull of ellipsoids is leveraged, inspired by \cite{nguyen2022convex}. Compared to \cite{nguyen2022convex}, we extend the invariant sets to $\lambda$-contractive sets and provide more insight into how the trajectories traverse through the ellipsoids, which will be leveraged in the subsequent sections for data-based control design. The following problem formalizes finding the maximum safe set for an open-loop system. 

\begin{problem}\label{Problem_2}
\textbf{(Largest CSiE using the convex hull of ellipsoids):} Consider the following open-loop deterministic LTI system
\begin{equation}\label{eq.LTI_openLoop}
x(t+1)=Ax(t).   
\end{equation}

Let $\mathcal{E}(P_i,1)$ for $i=1,\ldots,n_v$ be a set of ellipsoids, where $n_v$ denotes the number of ellipsoids. Find the largest safe set within the polyhedral admissible set $\mathcal{S}$ defined in \eqref{poly} using the convex hull of ellipsoids, i.e., $\mathcal{S}_c=\operatorname{Co}\big(\mathcal{E}(P_1,1),\ldots,\mathcal{E}(P_{n_v},1)\big)$.
\end{problem}

By considering the fact that if $x(t) \in \mathcal{S}_c$, then, according to \eqref{eq.combination}, it can be expressed as
\begin{equation}\label{eq.x_k_proof}
x(t)=\sum\limits_{i=1}^{n_v}\alpha_i(t)\upsilon_i(t),
\end{equation}
for some $\upsilon_i(t) \in \mathcal{E}(P_i,1)$, and the time-varying parameters $\alpha_i(t)$ satisfy the conditions $\sum\limits_{i=1}^{n_v}\alpha_i(t)=1$ and $0 \leq \alpha_i \leq 1$ for $i=1,\ldots,n_v$.


\begin{theorem}\label{theorem_1}
Consider the system \eqref{eq.LTI_openLoop}. Let there exist matrices $P_i \in$ $\mathbb{S}^{n}$ and positive scalars $\mu_i$ satisfying the following optimization problem for $i=1,\ldots,n_v$ and $j=\operatorname{mod}(i+n_v-2,n_v)+1$

\begin{align}\label{eq.optimization_openLoop}
& \max\limits_{P_i, \mu_i}\left\{\sum\limits_{i=1}^{n_v} \mu_i\right\}, \\
& \mathrm{s.t.} \nonumber \\
& \begin{bmatrix}\label{eq.openLoop_condition_1}
P_i & AP_j \\
(*) & \lambda P_j
\end{bmatrix} \succeq 0, \,\,\, \forall i=1,\ldots,n_v, \\
& \begin{bmatrix}\label{eq.openLoop_condition_3}
P_i & P_iF_l^T \\
(*) & g_l^2
\end{bmatrix} \succeq 0,  \,\,\, \forall i=1,\ldots,n_v, \,\,\, \forall l=1,\ldots,q, \\
& \begin{bmatrix}\label{eq.openLoop_condition_5}
1 & \mu_i d_i^T \\
(*) & P_i
\end{bmatrix} \succeq 0, \,\,\, \forall i=1,\ldots,n_v.
\end{align}

Then, $\mathcal{S}_c=\operatorname{Co}\big(\mathcal{E}(P_1,1),\ldots,\mathcal{E}(P_{n_v},1)\big)$ represents the largest $\lambda$-contractive subset of the admissible set $\mathcal{S}$ for the system \eqref{eq.LTI_openLoop}. $d_i \in \mathbb{R}^n$ represent the reference direction for the $i$-th ellipsoid for $i=1,\ldots,n_v$.
\end{theorem}

\begin{proof}
Inspired by \cite{nguyen2022convex}, one needs to show that if $x(t) \in \mathcal{S}_c$ then $x(t+1) \in \lambda \mathcal{S}_c$. Since the current state, i.e., $x(t)$, belongs to the convex hull of ellipsoids, it can be written as \eqref{eq.x_k_proof}. Now, according to \eqref{eq.x_k_proof}, if it is shown that $\upsilon_j(t) \in \mathcal{E}(P_j,1)$ leads to $\upsilon_j(t+1) \in \lambda\mathcal{S}_c$, then the proof is complete. To do so, by pre and post multiplying \eqref{eq.openLoop_condition_1} with
\begin{equation}\label{eq.temp_matrix_1_proof_1}
\begin{bmatrix}
I & 0 \\
0 & P_j^{-1}
\end{bmatrix},
\end{equation}
one gets
\begin{equation}\label{eq.openLoop_condition_1_changed}
\begin{bmatrix}
P_i & A \\
(*) & \lambda P_j^{-1}
\end{bmatrix} \succeq 0, \,\,\, \forall i=1,\ldots,n_v.
\end{equation}

Multiplying \eqref{eq.openLoop_condition_1_changed} by $\alpha_i(t)$ and summing them result in
\begin{equation}\label{eq.openLoop_condition_1_changed_2}
\begin{bmatrix}
\sum\limits_{i=1}^{n_v}\alpha_i(t)P_i & A \\
(*) & \lambda P_j^{-1}
\end{bmatrix} \succeq 0, \,\,\, \forall i=1,\ldots,n_v.
\end{equation}

In terms of the Schur complement, equation \eqref{eq.openLoop_condition_1_changed_2} is rewritten as
\begin{equation}\label{eq.Schur_proof_0}
A^T\big(\sum\limits_{i=1}^{n_v}\alpha_i(t)P_i\big)^{-1}A \leq \lambda P_j^{-1}.
\end{equation}

Now, due to the fact that $\upsilon_j(t+1)=A\upsilon_j(t)$, multiplying $\upsilon_j(t)$ and $\upsilon_j^T(t)$ on the right and left side of \eqref{eq.Schur_proof_0}, respectively, yields
\begin{equation}\label{eq.Schur_proof_0_2}
\upsilon_j^T(t+1)\big(\sum\limits_{i=1}^{n_v}\alpha_i(t)P_i\big)^{-1}\upsilon_j(t+1) \leq \lambda \upsilon_j^T(t)P_j^{-1}\upsilon_j(t),
\end{equation}
meaning that $\upsilon_j(t) \in \mathcal{E}(P_j,1)$ results in $\upsilon_j(t+1) \in \mathcal{E}(\sum\limits_{i=1}^{n_v}\alpha_i(t)P_i,\lambda)$.

We now show that $\mathcal{E}(\sum\limits_{i=1}^{n_v}\alpha_i(t)P_i,\lambda) \subseteq \lambda\mathcal{S}_c$. To do so, we will use a proof by contradiction. Let's assume the existence of a point $x_p$ in $\mathcal{E}(\sum\limits_{i=1}^{n_v}\alpha_i(t)P_i,\lambda)$ that is not within the convex hull of the ellipsoids. Without loss of generality, we can assume that $x_p$ lies on the boundary of $\mathcal{E}(\sum\limits_{i=1}^{n_v}\alpha_i(t)P_i,\lambda)$. Let $a_p \in \mathbb{R}^n$ be the supporting hyperplane of the set $\mathcal{E}(\sum\limits_{i=1}^{n_v}\alpha_i(t)P_i,\lambda)$ at the point $x_p$. Since both sets, $\mathcal{E}(\sum\limits_{i=1}^{n_v}\alpha_i(t)P_i,\lambda)$ and $\lambda\mathcal{S}_c$, are symmetric with respect to the origin, we have the following relationship

\begin{equation}\label{eq.proof_supporting_hyperplane}
|a_p^T x| < |a_p^T x_p|=b_p^2, \,\,\, \forall x \in \lambda\mathcal{S}_c
\end{equation}

Consequently, one has
\begin{equation}\label{eq.proof_supporting_hyperplane_2}
a_p^T \mathcal{E}(\sum\limits_{i=1}^{n_v}\alpha_i(t)P_i,\lambda) a_p = b_p^2
\end{equation}

Furthermore, based on \eqref{eq.proof_supporting_hyperplane}, we have
\begin{equation}\label{eq.proof_supporting_hyperplane_3}
|a_p^T x| < b_p^2, \,\,\, \forall x \in \lambda\mathcal{S}_c
\end{equation}

This inequality holds if and only if \cite{nguyen2020optimizing}
\begin{equation}\label{eq.proof_supporting_hyperplane_4}
a_p^T \lambda P_i a_p < b_p^2
\end{equation}

Hence, for all $\alpha_i(t) \geq 0$ and $\sum\limits_{i=1}^{n_v}\alpha_i(t)=1$, one has
\begin{equation}\label{eq.proof_supporting_hyperplane_5}
a_p^T \mathcal{E}(\sum\limits_{i=1}^{n_v}\alpha_i(t)P_i,\lambda) a_p < b_p^2
\end{equation}

This contradicts \eqref{eq.proof_supporting_hyperplane_2}. Therefore, we conclude that $\mathcal{E}(\sum\limits_{i=1}^{n_v}\alpha_i(t)P_i,\lambda) \subseteq \lambda\mathcal{S}_c$, and in accordance with \eqref{eq.Schur_proof_0_2}, this implies $\upsilon_j(t+1) \in \lambda\mathcal{S}_c$, or equivalently, $x(t+1) \in \lambda\mathcal{S}_c$.

We now provide conditions for inclusion of the contractive set inside the safe set. The ellipsoid $\mathcal{E}(P_i,1)$ is contained in the polytope $\mathcal{S}$ if and only if \cite{boyd1994linear}
\begin{equation}\label{eq.set_containment}
\max \{F_lx | x \in \mathcal{E}(P_i,1) \} \leq g_l,
\end{equation}
which is equivalent to \eqref{eq.openLoop_condition_3}. This is due to the fact that since the ellipsoidal set is within the polytope, one has
\begin{equation}\label{eq.set_containment_3}
F_lx \leq g_l.
\end{equation}
Multiplying \eqref{eq.set_containment_3} with its transpose gives
\begin{equation}\label{eq.set_containment_4}
F_lxx^TF_l^T \leq g_l^2.
\end{equation}
Also, according to the definition of ellipsoidal sets, one has $xx^T \leq P_i$. Thus, inequality \eqref{eq.set_containment_4} is equivalent to
\begin{equation}\label{eq.set_containment_5}
F_lP_iF_l^T \leq g_l^2.
\end{equation}
Now, applying the Schur complement to \eqref{eq.set_containment_5} results in the constraint \eqref{eq.openLoop_condition_3}.

Furthermore, to determine the largest convex hull among ellipsoids, the typical approach is to maximize the volume of the corresponding ellipsoids. Alternatively, one can aim to maximize the shape of the ellipsoids concerning specific reference directions or sets, as mentioned in \cite{hu2002analysis}. In this context, we will discuss optimizing the set with respect to a reference direction.

Let $d_i \in \mathbb{R}^n$ represent a reference direction for the ellipsoid $\mathcal{E}(P_i,1)$. The problem of optimizing $\mathcal{E}(P_i,1)$ with respect to $d_i$ is equivalent to maximizing $\mu_i$ under the constraint $\mu_i^2 d_i^T P_i^{-1}d_i \leq 1$ which, using the Schur complement, can be reformulated as \eqref{eq.openLoop_condition_5}. This completes the proof.
\end{proof}

The next proposition provides an insight into the proof of Theorem 1, and will be leveraged in probabilistic data-based control design.

\begin{proposition}\label{proposition_1}
Let the optimization problem \eqref{eq.optimization_openLoop}--\eqref{eq.openLoop_condition_5} be feasible for the open-loop system \eqref{eq.LTI_openLoop}. Also, let $x(t)$ be represented by \eqref{eq.x_k_proof}. Then, after every time-step, $\upsilon_i(t)$ traverses from one ellipsoid to its neighboring ellipsoid. That is, it shows a cyclic behavior w.r.t ellipsoids over time. partitioning of the obtained convex hull of ellipsoids.
\end{proposition}

\begin{proof}
According to the proof of Theorem 1, an interesting result is achieved. Applying the Schur complement to equation \eqref{eq.openLoop_condition_1_changed} gives
\begin{equation}\label{eq.Schur_proof_1}
A^TP_i^{-1}A \leq \lambda P_j^{-1}.
\end{equation}

Now, due to the fact that $\upsilon_j(t+1)=A\upsilon_j(t)$, multiplying $\upsilon_j(t)$ and $\upsilon_j^T(t)$ on the right and left side of \eqref{eq.Schur_proof_1}, respectively, yields
\begin{equation}\label{eq.Schur_proof_1_2}
\upsilon_j^T(t+1)P_i^{-1}\upsilon_j(t+1) \leq \lambda \upsilon_j^T(t)P_j^{-1}\upsilon_j(t),
\end{equation}
meaning that $\upsilon_j(t) \in \mathcal{E}(P_j,1)$ results in $\upsilon_j(t+1) \in \mathcal{E}(P_i,\lambda)$ for $i=1,\ldots,n_v$ and $j=\operatorname{mod}(i+n_v-2,n_v)+1$.
\end{proof}

Illustrative explanation of this proposition is exhibited in figure \eqref{fig.proof}. This figure shows that both extreme points, denoted as $\upsilon_1(t)$ and $\upsilon_2(t)$, exhibit a cyclical movement between the level sets of two ellipsoids over time. Specifically, at time-step $t$, $\upsilon_1(t)$ resides within the boundaries of the blue ellipsoid. Subsequently, at time-step $t+1$, it transitions into the level set of the red ellipsoid, and then at time-step $t+2$, it enters the level set of the blue ellipsoid. This cyclic pattern persists until all extreme points converge to the origin, consequently leading to the convergence of the system trajectory towards the origin.

\begin{figure}
    \centering
    \includegraphics[width=0.5\columnwidth]{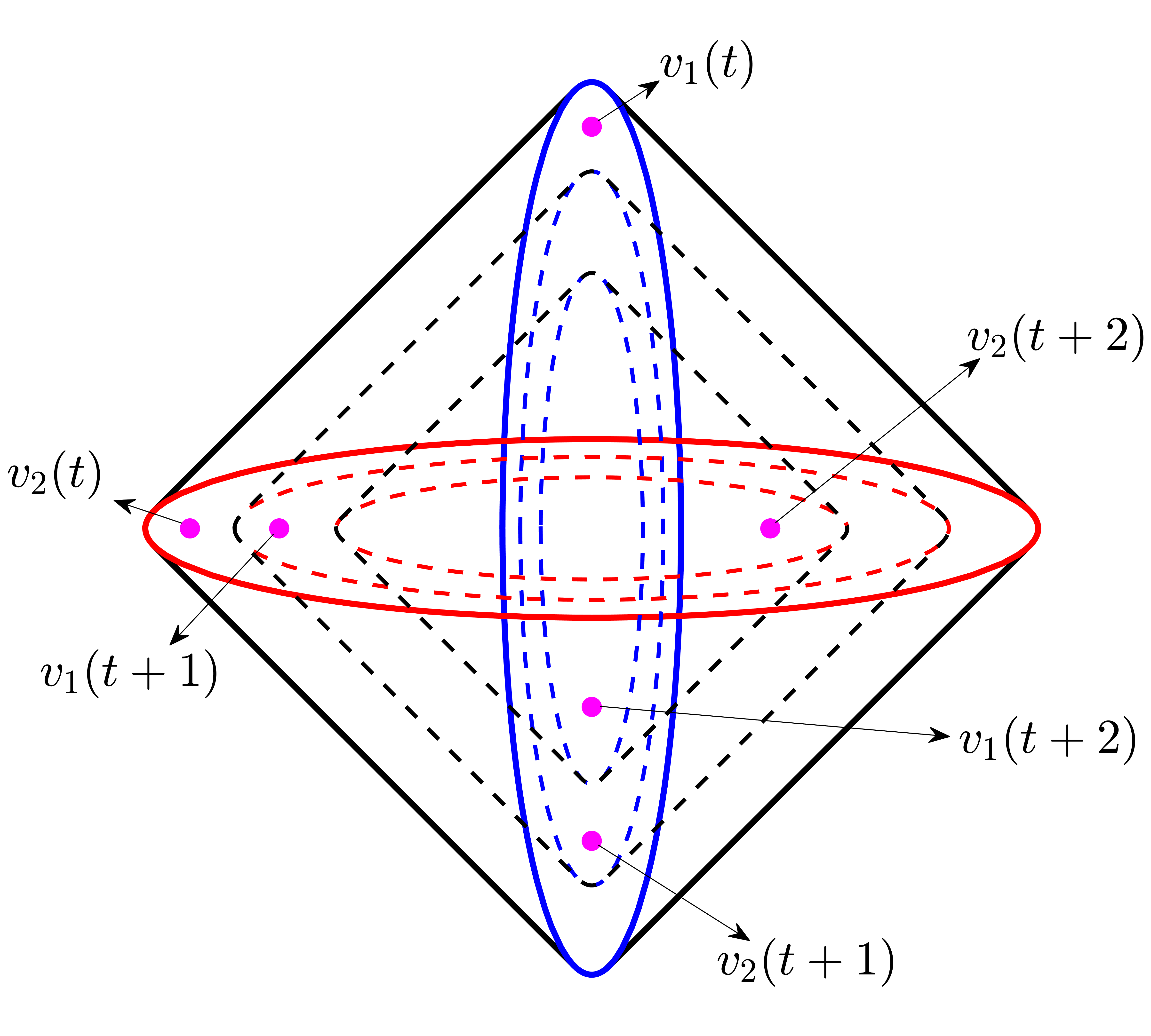}
    \caption{Illustrative diagram for the proof of Theorem 1.}
    \label{fig.proof}
\end{figure}




\begin{remark}\label{Remark_0}
If one considers only the case $i = j$ in equation \eqref{eq.openLoop_condition_1}, it implies that each ellipsoid must remain invariant by itself. However, this approach has two drawbacks:

1) The optimization problem presented in \eqref{eq.optimization_openLoop} must be solved $n_v$ times, leading to an increase in the computational cost of the control method. Additionally, the generated ellipsoids only differ in their orientations. In essence, this is equivalent to rotating a single ellipsoid toward different vertices of a polyhedral set. Based on our simulations, it is highly likely that all the generated ellipsoids will become identical and will not cover a significant portion of the admissible set.

2) The primary objective of safety backup controllers is to minimize interference with the optimal policy in order to maintain the system's optimal performance. However, if each ellipsoid is forced to remain invariant, the freedom of the system trajectory is severely restricted once the system states enter one of the ellipsoids. This is because the invariance property of ellipsoids prevents the trajectory from evolving within the convex hull and confines it to a specific ellipsoid. To address these issues, this paper considers the more general condition presented in equation \eqref{eq.openLoop_condition_1} and relaxes the requirement for ellipsoids to be invariant.
\end{remark}

\section{Probabilistic Safety Backup Policy Design: Model-based Approach}
In this section, we introduce a technique for creating a model-driven solution to design a safety backup policy for Problem 1. The presented method defines conditions to identify and generate the largest ellipsoidal sets, such that their convex hull is the maximum subset of the primary polyhedral admissible set of the system \eqref{eq.LTI}, ensuring $\lambda$-contractiveness. The provided theorem outlines these conditions, which guarantee that the probabilistic behavior of the system remains within a scaled version of the convex hull of the ellipsoids. By satisfying these conditions, the model-based policy can ensure both safety and stability, even in the presence of external factors such as noise.

\begin{problem}\label{Problem_3}
\textbf{(Largest CSiE for the closed-loop system using the convex hull of ellipsoids):} Consider the LTI system \eqref{eq.LTI} under Assumptions 1--3. Also, consider the admissible set $\mathcal{S}$. Design partitions $\mathcal{C}_1,\ldots,\mathcal{C}_{N_p}$ and a piecewise-affine controller in the form of
\begin{equation}\label{eq.general_controller}
u(t)=
\left\{\begin{array}{cc}
K_1^p x(t) \,\,\, \mathrm{if} \,\,\, x(t) \in \mathcal{C}_1 \\
\vdots \\
K_{N_p}^p x(t) \,\,\, \mathrm{if} \,\,\, x(t) \in \mathcal{C}_{N_p}
\end{array}\right.
\end{equation}
to maximize the size of $\mathcal{S}_c = \{\cup_{i=1}^{N_p} \mathcal{C}_i\} \subseteq \mathcal{S}$ such that $\mathcal{S}_c$ is CSiE for the closed-loop system, where $N_p$ denotes the number of partitions of the convex hull of ellipsoids.
\end{problem}


The number and boundaries of the piecewise-affine regions are determined by the ellipsoids used to construct the convex hull \( \mathcal{S}_c \); specifically, Algorithm~1 in Section~7 provides a systematic vertex-extraction and partitioning procedure based on solving a set of ellipsoidal boundary equations, followed by convex hull computation using the Quickhull algorithm \cite{barber1996quickhull}. Each region is then defined by the set of extreme points (including the origin) associated with neighboring ellipsoids, and the partitions emerge automatically without manual tuning.


Our approach focuses on utilizing the convex hull of ellipsoids as a foundational concept. Initially, using an optimization algorithm, we compute a state-feedback gain for each ellipsoid. Subsequently, in the next theorem, we design model-based state-feedback controllers with an emphasis on expectation. Following this, we present the data-based counterpart of Theorem 2 in terms of expectation (Theorem 3) and in terms of probability (Theorem 4). We then elaborate on the process of partitioning the derived convex hull and explain how to compute a state-feedback controller for each of these partitions. It is worth noting that since both the partitioning procedure and the computation of state-feedback controllers are applicable to both model-based and data-based scenarios, we present these aspects in Section VI for the sake of coherence.

\begin{theorem}\label{theorem_2}
Consider the system \eqref{eq.LTI} that satisfies assumptions 1--3. Let there exist matrices $P_i \in$ $\mathbb{S}^{n}$ and $S_i \succeq 0$, and positive scalars $\mu_i$ for $i=1,\ldots,n_v$ such that the following optimization problem is feasible

\begin{align}\label{eq.optimization_closedLoop}
& \max\limits_{P_i, \mu_i, S_i}\left\{\sum\limits_{i=1}^{n_v} \mu_i\right\}, \\
& \mathrm{s.t.} \nonumber \\
& \begin{bmatrix}\label{eq.closedLoop_condition_1}
P_i & AP_j+BS_j \\
(*) & \lambda P_j
\end{bmatrix} \succeq 0, \,\,\, \forall i=1,\ldots,n_v, \\
& \begin{bmatrix}\label{eq.closedLoop_condition_3}
P_i & P_iF_l^T \\
(*) & g_l^2
\end{bmatrix} \succeq 0,  \,\,\, \forall i=1,\ldots,n_v, \,\,\, \forall l=1,\ldots,q, \\
& \begin{bmatrix}\label{eq.closedLoop_condition_5}
1 & \mu_i d_i^T \\
(*) & P_i
\end{bmatrix} \succeq 0, \,\,\, \forall i=1,\ldots,n_v.
\end{align}

Then, $\mathcal{S}_c=\operatorname{Co}\big(\mathcal{E}(P_1,1),\ldots,\mathcal{E}(P_{n_v},1)\big)$ represents the largest CSiE subset of the admissible set $\mathcal{S}$ for the closed-loop system \eqref{eq.LTI}, and the controller gains are computed as $K_i=S_iP_i^{-1}$. Also, index $j$ is $j=\operatorname{mod}(i+n_v-2,n_v)+1$ for $i=1,\ldots,n_v$.
\end{theorem}

\begin{proof}
It has to be shown that for $x(t) \in \mathcal{S}_c$, there exists a controller $u(t)$ such that $x(t+1) \in \lambda\mathcal{S}_c$. $x(t)$ is decomposed as \eqref{eq.x_k_proof}. Consider the following control law
\begin{equation}\label{eq.control_law}
u(t)=\sum\limits_{i=1}^{n_v}\alpha_i(t)u_i(t),
\end{equation}
where $\alpha_i(t)$ have been defined in \eqref{eq.x_k_proof}, and
\begin{equation}\label{eq.control_law_i}
u_i(t)=K_ix(t).
\end{equation}

Substituting \eqref{eq.control_law}, \eqref{eq.control_law_i}, and \eqref{eq.x_k_proof} into the expectation of next state, i.e., $\mathbb{E}[x(t+1)]=Ax(t)+Bu(t)$, results
\begin{align}\label{eq.x_k_2_proof}
\mathbb{E}[x(t+1)] = \; & \sum\limits_{i=1}^{n_v}\alpha_i(t)(A+BK_i)\upsilon_i(t) \nonumber \\
= \; & \sum\limits_{i=1}^{n_v}\alpha_i(t)\upsilon_i(t+1),
\end{align}
with 
\begin{equation}\label{eq.proof_2_temp}
\mathbb{E}[\upsilon_i(t+1)]=(A+BK_i)\upsilon_i(t).
\end{equation}

Now, by using \eqref{eq.closedLoop_condition_1}, it is shown that if $\upsilon_j(t) \in \mathcal{E}(P_j,1)$ then $\upsilon_j(t+1) \in \lambda\mathcal{S}_c$. Define $S_i=K_iP_i$. Hence, the condition \eqref{eq.closedLoop_condition_1} becomes
\begin{equation}\label{eq.closedLoop_condition_1_changed}
\begin{bmatrix}
P_i & (A+BK_j)P_j \\
(*) & \lambda P_j
\end{bmatrix} \succeq 0, \,\,\, \forall i=1,\ldots,n_v
\end{equation}

By pre and post multiplying \eqref{eq.closedLoop_condition_1_changed} with \eqref{eq.temp_matrix_1_proof_1}, one obtains
\begin{equation}\label{eq.closedLoop_condition_1_changed_2}
\begin{bmatrix}
P_i & (A+BK_j) \\
(*) & \lambda P_j^{-1}
\end{bmatrix} \succeq 0, \,\,\, \forall i=1,\ldots,n_v.
\end{equation}

Multiplying \eqref{eq.closedLoop_condition_1_changed_2} by $\alpha_i(t)$ and summing them result in
\begin{equation}\label{eq.closedLoop_condition_1_changed_3}
\begin{bmatrix}
\sum\limits_{i=1}^{n_v}\alpha_i(t)P_i & (A+BK_j) \\
(*) & \lambda P_j^{-1}
\end{bmatrix} \succeq 0, \,\,\, \forall i=1,\ldots,n_v.
\end{equation}

In terms of the Schur complement, equation \eqref{eq.closedLoop_condition_1_changed_3} is rewritten as
\begin{equation}\label{eq.Schur_proof_2}
(A+BK_j)^T\big(\sum\limits_{i=1}^{n_v}\alpha_i(t)P_i\big)^{-1}(A+BK_j) \leq \lambda P_j^{-1}.
\end{equation}

Now, due to the fact that $\mathbb{E}[\upsilon_j(t+1)]=(A+BK_j)\upsilon_j(t)$, multiplying $\upsilon_j(t)$ and $\upsilon_j^T(t)$ on the right and left side of \eqref{eq.Schur_proof_2}, respectively, yields
\begin{equation}\label{eq.Schur_proof_2_2}
\upsilon_j^T(t+1)\big(\sum\limits_{i=1}^{n_v}\alpha_i(t)P_i\big)^{-1}\upsilon_j(t+1) \leq \lambda \upsilon_j^T(t)P_j^{-1}\upsilon_j(t).
\end{equation}

The rest of the proof is analogous to that of Theorem 1 and is omitted here.
\end{proof}

\section{Probabilistic Safety Backup Policy Design: Data-based Approach}
The purpose of this section is to introduce a data-driven alternative to condition \eqref{eq.closedLoop_condition_1} that removes the requirement for a system model in the safe-controller. Initially, a certainty equivalence-based direct learning technique that enables the acquisition of a risk-neutral safety backup policy based on the definition 6 is developed. Subsequently, by leveraging the minimum-variance approach outlined in \cite{modares_MV}, a direct probabilistic learning version of the previous method is presented to guarantee that the convex hull of ellipsoids is CSiP. The aim is to decrease the variance of the closed-loop system with respect to the safe set generated by the convex hull of ellipsoids and mitigate the risk of safety violations in noisy environments.

To accomplish this, let us begin by assuming that an input sequence of $u(0),u(1),\ldots,u(N-1)$ is applied to the system \eqref{eq.LTI}, and $N$ samples of states are collected. Subsequently, these samples are organized in the following manner:
\begin{align}
& U_0 = [u(0),u(1),\ldots,u(N-1)], \label{eq.Data_U_0} \\
& X_0 = [x(0),x(1),\ldots,x(N-1)], \label{eq.Data_X_0} \\
& X_1 = [x(1),x(2),\ldots,x(N)]. \label{eq.Data_X_1} 
\end{align}

Also, the noise sequence is as follows
\begin{equation}\label{eq.Data_W_0}
W_0 = [w(0),w(1),\ldots,w(N-1)].
\end{equation}

\begin{assumption}\label{Assumption_4}
The data matrix $X_0$ in \eqref{eq.Data_X_0} is full row rank, with sample count being at least $n+1$.
\end{assumption}

\begin{remark}
For an indirect data-based control of the LTI system \eqref{eq.LTI}, which involves identifying the matrices $A$ and $B$, it is essential for the data matrix denoted by
\begin{align} \label{eq.rank}
\left[ \begin{array}{l}
{U_0} \\
{X_0}
\end{array} \right]
\end{align}
to possess a full row rank. However, when aiming to directly learn a safe controller, as shown later, only Assumption 4 is needed, which requires a smaller number of samples as it is only necessary to ensure that the matrix $X_0$ possesses a full row rank.
\end{remark}

Subsequently, the collected data are utilized to derive data-driven versions of condition \eqref{eq.closedLoop_condition_1} from both risk-neutral and risk-aware perspectives. The resulting condition can be directly employed in designing a safe control policy, eliminating the need for the system model.

The first step is to provide a data-based representation of the closed-loop system. Hence, inspired by \cite{bisoffi2020data}, based on the data collected in \eqref{eq.Data_U_0}--\eqref{eq.Data_X_1} and the stochastic linear system \eqref{eq.LTI}, one has
\begin{align}\label{eq.LTI_DD_form}
X_1-W_0 = A X_0 + BU_0.
\end{align}

According to Assumption 3, there exists a right inverse for $X_0$ such that
\begin{align}\label{eq.G_K}
 X_0 G_K=I
\end{align}

Thus, multiplying both sides of \eqref{eq.LTI_DD_form} by $G_{K}$ from right yields
\begin{align}\label{eq.LTI_DD_form_2}
(X_1-W_0)G_K = A + B U_0G_K.
\end{align}

By defining the controller gain as $K=U_0G_K$, the closed-loop system can be written as
\begin{align}\label{eq.LTI_DD_form_3}
A + BK = (X_1-W_0)G_K.
\end{align}

Hence,
\begin{align}\label{eq.LTI_DD_form_4}
x(k+1)=(X_1-W_0)G_Kx(k)+w(k).
\end{align}





\begin{problem}\label{Problem_4}
\textbf{(Data-based safe control design with largest CSiP inside the admissible set):} Consider the LTI system \eqref{eq.LTI} under Assumptions 1--4. Let $\mathcal{E}(P_i,1)$ for $i=1,\ldots,n_v$ be a set of ellipsoids. Find the largest CSiP  within the admissible set $\mathcal{S}$ by designing data-based state-feedback controllers in the form of $u_i(t)=K_ix(t)$ for $i=1,\ldots,n_v$, and the piecewise-affine safe controller as defined in \eqref{eq.general_controller}: \\
\textbf{I.} First, by assuming that noise is measurable. \\
\textbf{II.} Second, by relaxing the noise measurement assumption.
\end{problem}

\subsection{Certainty Equivalence Perspective}
In this subsection, a data-driven risk-neutral certainty-equivalence direct learning method is introduced. It aims to acquire a state-feedback gain within each ellipsoid that generates the convex hull. After establishing the following hypothesis, the results of this method are condensed in the subsequent theorem.

\begin{assumption}\label{Assumption_5}
The noise sequence $w(k)$ can be measured and collected as a data matrix for $N$ samples, as shown in \eqref{eq.Data_W_0}.
\end{assumption}

\begin{theorem}\label{theorem_3}
Consider the system \eqref{eq.LTI} that satisfies Assumptions 1--5. Data are collected and arranged as equations \eqref{eq.Data_U_0}--\eqref{eq.Data_X_1}. Let there exist matrices $P_i \in$ $\mathbb{S}^{n}$ and $Y_i \succeq 0$, and positive scalars $\mu_i$ for $i=1,\ldots,n_v$ such that the following optimization problem is feasible

\begin{align}\label{eq.optimization_closedLoop_DD}
& \max\limits_{P_i, Y_i, \mu_i}\left\{\sum\limits_{i=1}^{n_v} \mu_i\right\}, \\
& \mathrm{ s.t.} \nonumber \\
& \begin{bmatrix}\label{eq.closedLoop_condition_DD_1}
P_i & (X_1-W_0) Y_j \\
(*) & \lambda P_j
\end{bmatrix} \succeq 0, \,\,\, \forall i=1,\ldots,n_v, \\
& \begin{bmatrix}\label{eq.closedLoop_condition_DD_3}
P_i & P_iF_l^T \\
(*) & g_l^2
\end{bmatrix} \succeq 0,  \,\,\, \forall i=1,\ldots,n_v, \,\,\, \forall l=1,\ldots,q, \\
& \begin{bmatrix}\label{eq.closedLoop_condition_DD_5}
1 & \mu_i d_i^T \\
(*) & P_i
\end{bmatrix} \succeq 0, \,\,\, \forall i=1,\ldots,n_v, \\
& X_0 Y_i = P_i, \,\,\, \forall i=1,\ldots,n_v. \label{eq.closedLoop_condition_DD_6}
\end{align}

Then, $\mathcal{S}_c=\operatorname{Co}\big(\mathcal{E}(P_1,1),\ldots,\mathcal{E}(P_{n_v},1)\big)$ represents the largest CSiE subset of the admissible set $\mathcal{S}$ for the closed-loop system \eqref{eq.LTI}. Moreover, the controller gains for ellipsoids are calculated as $K_i=U_0Y_iP_i^{-1}$. Also, index $j$ is computed as $j=\operatorname{mod}(i+n_v-2,n_v)+1$ for $i=1,\ldots,n_v$.
\end{theorem}

\begin{proof}
We show that the constraints \eqref{eq.closedLoop_condition_DD_1} and \eqref{eq.closedLoop_condition_DD_6} together provide an equivalent data-based form of the constraint \eqref{eq.closedLoop_condition_1} in Theorem 2. This concludes the equivalence of (58)-(62) and (35)-(38), as other constraints are common in the two optimizations. We first provide a data-based representation of the closed-loop systems obtained by control gains of each ellipsoids. Since there is a decision variable for every ellipsoids (corresponding to every control gain $K_i$), data-based representations \eqref{eq.G_K} and \eqref{eq.LTI_DD_form_3} amount to $X_0G_{K,i}=I$ and $A+BK_i=(X_1-W_0)G_{K,i}$ for the $i$-th ellipsoid, with $K_i=U_0 G_{K,i}$. Furthermore, since the rank of $X_0$ is $n$, the right inverse $G_{K,i}$ exists, and since at least $n+1$ samples are collected, $G_{K,i}$ is not unique and thus it can be considered as a decision variable. 

To show the equivalence of \eqref{eq.closedLoop_condition_DD_1} and \eqref{eq.closedLoop_condition_DD_6} with \eqref{eq.closedLoop_condition_1}, define first $Y_j=G_{K,j}P_j$. Then, under \eqref{eq.closedLoop_condition_DD_6}, one has $(X_1-W_0)Y_j=(A+BK_j)P_j$, and thus the 
constraints \eqref{eq.closedLoop_condition_DD_1} and \eqref{eq.closedLoop_condition_DD_6} become
\begin{align}\label{eq.closedLoop_condition_DD_1_changed}
\begin{bmatrix}
P_i & (A+BK_j)P_j \\
(*) & \lambda P_j
\end{bmatrix} \succeq 0, \,\,\, \forall i=1,\ldots,n_v.
\end{align}
 which is transformed to \eqref{eq.closedLoop_condition_1} using $S_j=K_jP_j$ defined in the poof of Theorem 2.  This completes the proof. 
 
\end{proof}

The safety backup controller learned according to Theorem 3 suffers from a drawback in that it necessitates the measurement of noise, which is not feasible in practice. To overcome this challenge, a data-driven safe controller based on minimum variance is designed in the subsequent part of the paper. The objective of the designed controller in the upcoming subsection is to eliminate the requirement for noise measurement and enhance the practical applicability of the safe controller. This approach involves collecting data from the system and utilizing this data to establish conditions that not only compute the controller gains but also minimize the variance of the closed-loop system. By utilizing this approach, a state-feedback controller is constructed for each of the ellipsoids forming the convex hull, guaranteeing the stability of the closed-loop system without requiring noise measurement.

\subsection{Probabilistic Perspective}
In this subsection, a minimum variance-based approach is presented to alleviate the restrictive assumption related to the availability of noise measurements. The objective of this approach is to address the limitations associated with this assumption considered in \cite{modares2023}. In conventional indirect learning methods \cite{krishnan2021direct}, predetermined high-confidence sets are assigned to the dynamics $A$ and $B$. Consequently, the controller gain $K$ can only impact the variance associated with the $BK$ component of the closed-loop dynamics. In contrast, with the proposed minimum variance-based direct learning approach, the entire closed-loop dynamics $A+BK$ is learned, and the control gain $K$ can be designed to decrease the variance for the entire closed-loop dynamics. 

In addition to learning the closed-loop dynamics directly, the proposed approach minimizes the variance of the state distribution to ensure high-probability safety guarantees. Instead of treating noise as bounded or unstructured, this variance-aware formulation explicitly shapes the distribution of the next state by designing control gains that reduce its spread. The resulting controller is designed to satisfy a probabilistic constraint of the form \( \mathbb{P}[x(t) \in \mathcal{S}_c] \geq 1 - \epsilon \), where \( \mathcal{S}_c \subseteq \mathcal{S} \) is a learned safe set. This ensures that constraint violations due to stochastic disturbances remain within an acceptable risk threshold.

The following Lemma is brought up for the $j$th ellipsoid's state, i.e., $\upsilon_j(t)$.

\begin{lemma}\label{lemma_3}
Consider the system \eqref{eq.LTI}. Let Assumptions 1--5 be satisfied. Let the controller be  $u_j(t)=K_j\upsilon_j(t)=U_0G_{K,j}\upsilon_j(t)$ for $j$th ellipsoid, where $X_0G_{K,j}=I$. Then, with probability $1-\delta$, the next state $\upsilon_j(t+1)$ is steered into the following confidence ellipsoid
\begin{align}\label{eq.confidence_ellips}
& \mathcal{E}(V_j,1) = \nonumber \\
& \Big\{\upsilon_j:\big(\upsilon_j-X_1G_{K,j}\upsilon_j(t)\big)^T V_j^{-1} \big(\upsilon_j-X_1G_{K,j}\upsilon_j(t)\big) \leq 1 \Big\},
\end{align}
where  
\begin{align}\label{eq.variance}
V_j=\bigg(n+2 \sqrt{n\,\, \operatorname{log}{\frac{1}{\delta}}}+2\operatorname{log}{\frac{1}{\delta}}\bigg) \bigg(\operatorname{Tr}\big(G_{K,j}P_j^{-1}G_{K,j}^T\big)\Sigma+\Sigma\bigg).
\end{align} 
\end{lemma}

\begin{proof}
Similar to the proof of Theorem 3 and based on the general data-based model \eqref{eq.LTI_DD_form_4}, for the $j$th ellipsoid, one has
\begin{equation}\label{eq.ellips_state}
 \upsilon_j(t+1)=X_1G_{K,j}\upsilon_j(t)-W_0G_{K,j}\upsilon_j(t)+w(t).   
\end{equation}

Based on \eqref{eq.LTI_DD_form_3}, the nominal model of $A+BK_j$ is $X_1G_{K,j}$. Now define the nominal next state in the $j$th ellipsoid as
\begin{align}\label{eq.ellips_state_nominal}
\bar{\upsilon}_j(t+1)=X_1G_{K,j}\upsilon_j(t).
\end{align}

Then, for the random variable $\tilde{\upsilon}_j(t+1)=\upsilon_j(t+1)-\bar{\upsilon}_j(t+1)=-W_0G_{K,j}\upsilon_j(t)+w(t)$, its covariance satisfies
\begin{align}\label{eq.ellips_state_variance}
& \mathbb{E}[\tilde{\upsilon}_j(t+1)\tilde{\upsilon}_j(t+1)^T]=\mathbb{E}\Big[W_0G_{K_j}\upsilon_j(t)\upsilon_j(t)^TG_{K,j}^TW_0^T\Big]+\Sigma,
\end{align}
which is concluded by using \eqref{eq.Variance_1_2}. Furthermore, since $\upsilon_j(t)^T P_j^{-1} \upsilon_j(t) \leq 1$, using the Schur complement, one gets $\upsilon_j(t)\upsilon_j(t)^T \leq P_j$. Thus,
\begin{align}\label{eq.ellips_state_variance_2}
& \mathbb{E}[\tilde{\upsilon}_j(t+1)\tilde{\upsilon}_j(t+1)^T] \leq \mathbb{E}\Big[W_0G_{K,j}P_jG_{k,j}^T W_0^T\Big]+\Sigma \nonumber \\
& =\operatorname{Tr}(G_{K,j}P_jG_{K,j}^T)\Sigma+\Sigma=\bar{V}_j.
\end{align}

Therefore, since also $\mathbb{E}[\tilde{\upsilon}_j(t+1)]=0$, $\tilde{\upsilon}_j(t+1)$ is a sub-Gaussian random vector with covariance $\bar{V}_j$. Thus, with probability at least $1-\delta$, one has \cite{lattimore2020bandit}
\begin{align}\label{eq.confidence_ellips_2}
\tilde{\upsilon}_j(t+1)^T \bar{V}_j^{-1} \tilde{\upsilon}_j(t+1) \leq n+2 \sqrt{n\operatorname{log}{\frac{1}{\delta}}}+2\operatorname{log}{\frac{1}{\delta}}=\delta_n.
\end{align}

Equivalently, with probability at least $1-\delta$, one has
\begin{align}
\tilde{\upsilon}_j(t+1)^T V_j^{-1} \tilde{\upsilon}_j(t+1) \leq 1.
\end{align}

This completes the proof.
\end{proof}

\begin{theorem}\label{theorem_4}
Consider the system \eqref{eq.LTI} that satisfies assumptions 1--4. Also, data are collected and arranged as equations \eqref{eq.Data_U_0}--\eqref{eq.Data_X_1}. Let there exist matrices $P_i \in$ $\mathbb{S}^{n}$, $Y_i \succeq 0$, and $H_i \succeq 0$, and positive scalars $\mu_i$, $\eta_i$, and $\zeta_i$ for $i=1,\ldots,n_v$ such that the following optimization problem is feasible for some $\tau_i$

\begin{align}\label{eq.optimization_closedLoop_DD}
& \max\limits_{P_i, Y_i, H_i, \mu_i, \eta_i, \zeta_i}\left\{\sum\limits_{i=1}^{n_v} (\mu_i-\eta_i-\zeta_i)\right\}, \\
& \mathrm{s.t.} \nonumber \\
& \begin{bmatrix}\label{eq.closedLoop_condition_DD_MV_1}
P_i & X_1Y_j & \eta_j\Sigma^{\frac{1}{2}} \\
(*) & (\lambda-\tau_j)P_j & 0 \\
(*) & (*) & \frac{\tau_j}{\delta_n}I
\end{bmatrix} \succeq 0, \,\,\, \forall i=1,\ldots,n_v \\
& \begin{bmatrix}\label{eq.closedLoop_condition_DD_MV_3}
P_i & P_iF_l^T \\
(*) & g_l^2
\end{bmatrix} \succeq 0,  \,\,\, \forall i=1,\ldots,n_v, \,\,\, \forall l=1,\ldots,q \\
& \begin{bmatrix}\label{eq.closedLoop_condition_DD_MV_5}
1 & \mu_i d_i^T \\
(*) & P_i
\end{bmatrix} \succeq 0, \,\,\, \forall i=1,\ldots,n_v \\
& X_0 Y_i = P_i, \,\,\, \forall i =1,\ldots,n_v \label{eq.closedLoop_condition_DD_MV_6} \\
& \begin{bmatrix}\label{eq.closedLoop_condition_DD_MV_7}
H_i & Y_i \\
(*) & P_i
\end{bmatrix} \succeq 0, \,\,\, \forall i=1,\ldots,n_v \\
& \begin{bmatrix}\label{eq.closedLoop_condition_DD_MV_8}
\zeta_i+1 & \eta_i \\
(*) & 1
\end{bmatrix} \succeq 0, \,\,\, \forall i=1,\ldots,n_v \\
& \mathrm{Tr}(H_i) \leq \zeta_i, \,\,\, \forall i=1,\ldots,n_v \label{eq.closedLoop_condition_DD_MV_9}
\end{align}

Then, $\mathcal{S}_c=\operatorname{Co}\big(\mathcal{E}(P_1,1),\ldots,\mathcal{E}(P_{n_v},1)\big)$ represents the largest CSiP subset of the admissible set $\mathcal{S}$ for the closed-loop system \eqref{eq.LTI} with the risk level $\delta$. Moreover, the controller gains for ellipsoids are calculated as $K_i=U_0Y_iP_i^{-1}$. Also, index $j$ is computed as $j=\operatorname{mod}(i+n_v-2,n_v)+1$ for $i=1,\ldots,n_v$.
\end{theorem}

\begin{proof}
First and foremost, according to Proposition 1, the CSiP property for the state, i.e., $x(t) \in \mathcal{S}_c \Rightarrow \mathbb{P}[x(t+1) \in \lambda\mathcal{S}_c] \geq (1-\delta)$ is equivalent to the following constraint
\begin{equation}\label{eq.equivalence_MV}
\upsilon_j(t) \in \mathcal{E}(P_j,1) \Rightarrow \mathbb{P}[\upsilon_j(t+1) \in \mathcal{E}(P_i,\lambda)] \geq (1-\delta),
\end{equation}
where $j=\operatorname{mod}(i+n_v-2,n_v)+1$ for $i=1,\ldots,n_v$. Probabilistic $\lambda$-contractivity with the risk level $\delta$ is satisfied if \eqref{eq.equivalence_MV} holds.
Satisfaction of the right-hand side amounts to assure that the set of possible next states with probability $1-\delta$ is a subset of the safe set. That is, based on Lemma 3,
\begin{align}
\mathbb{P}\big[\upsilon_j(t+1) \in \mathcal{E}(P_i,\lambda)\big] \geq (1-\delta),
\end{align}
is satisfied if
\begin{align}
& \Big \{\upsilon_j:\big(\upsilon_j-X_1G_{K,j}\upsilon_j(t)\big)^T V_j^{-1} \big(\upsilon_j-X_1G_{K,j}\upsilon_j(t)\big) \leq 1 \Big\} \nonumber \\ 
& \subseteq \Big \{\upsilon_j:\upsilon_j^T P_i^{-1}\upsilon_j \leq \lambda \Big \},
\end{align}

Equivalently,
\begin{align}
&  1-\big(\upsilon_j-X_1G_{K,j}\upsilon_j(t)\big)^T V_j^{-1} \big(\upsilon_j-X_1G_{K,j}\upsilon_j(t)\big) \geq 0 \nonumber \\  
& \Rightarrow \lambda-\upsilon_j^T P_i^{-1} \upsilon_j \geq 0.
\end{align}

Using the S-procedure, this is equivalent to the condition that there exists a $\tau_i$ that satisfies
\begin{align}\label{eq.S_procedure_1}
& \lambda-\upsilon_j^T P_i^{-1} \upsilon_j-\tau_j \times \nonumber \\
& \big[1-\big(\upsilon_j-X_1G_{K,j}\upsilon_j(t)\big)^T V_j^{-1} \big(\upsilon_j-X_1G_{K,j}\upsilon_j(t)\big) \big] \geq 0, \,\, \forall \upsilon_j
\end{align}
The $\upsilon_j$ that minimizes this expression is
\begin{align}
\upsilon_j=-\tau_j(P_i^{-1}-\tau_j V_j^{-1})^{-1} V_j^{-1} X_1 G_{k,j} \upsilon_j(t).
\end{align}
Replacing this expression into \eqref{eq.S_procedure_1} yields
\begin{align}\label{eq.S_procedure_2}
\lambda-\tau_j-\upsilon_j^T(t) G_{K,j}^T X_1^T \big(P_i-\frac{1}{\tau_j} V_j \big)^{-1} X_1 G_{K,j} \upsilon_j(t) \geq 0.
\end{align}

Using the Schur complement on \eqref{eq.S_procedure_2} gives
\begin{align}
\begin{bmatrix}
P_i-\frac{1}{\tau_j}V & X_1 G_{K,j}\upsilon_j(t)  \\
(*) & \lambda-\tau_j
\end{bmatrix} \succeq 0, \,\,\, \forall \upsilon_j(t) \in \mathcal{E}(P_j,1).
\end{align}

Using Schur complement again and using $V_j$ in \eqref{eq.variance} yields
\begin{align}\label{eq.Schur_2}
& P_i-\frac{1}{\lambda-\tau_j} X_1 G_{K,j}\upsilon_j(t)\upsilon_j^T(t) G_{K,j}^T X_1^T- \nonumber \\
& \frac{\delta_n}{\tau_j}\bigg(\operatorname{Tr}(G_{K,j} P_j G_{K,j}^T)\Sigma+\Sigma\bigg) \succeq 0, \,\,\, \forall \upsilon_j(t) \in \mathcal{E}(P_j,1).
\end{align}

Since $\upsilon_j(t) \in \mathcal{E}(P_j,1)$, one has $\upsilon_j(t) \upsilon_j(t)^T \leq P_j$. Therefore, a sufficient condition for the satisfaction of \eqref{eq.Schur_2} is
\begin{align}\label{eq.Schur_3}
& P_i-\frac{\delta_n}{\tau_j}\Sigma-\frac{1}{\lambda-\tau_j} X_1 G_{K,j}P_jG_{K,j}^T X_1^T- \nonumber \\ 
& \frac{\delta_n}{\tau_j}\operatorname{Tr}\big(G_{K,j}P_j G_{K,j}^T\big)\Sigma \succeq 0.
\end{align}

Defining $Y_j=G_{K,j}P_j$, one has 
\begin{align}\label{eq.Schur_4}
& P_i-\frac{\delta_n}{\tau_j}\Sigma-\frac{1}{\lambda-\tau_j}X_1 Y_j P_j^{-1} Y_j^T X_1^T-\frac{\delta_n}{\tau_j}
\operatorname{Tr}\big(Y_j P_j^{-1} Y_j^T\big)\Sigma \succeq 0.
\end{align}

A sufficient condition for the satisfaction of this inequality is
\begin{align}
& P_i-\frac{\delta_n \eta_j^2}{\tau_j}\Sigma-\frac{1}{\lambda-\tau_j} X_1 Y_j P_j^{-1} Y_j^T X_1^T \succeq 0, \label{eq.Schure_5_1} \\
& 1+\operatorname{Tr}\big(Y_j P_j^{-1} Y_j^T\big) \leq \eta_j^2. \label{eq.Schure_5_2}
\end{align}
for which $\eta_j$ can be minimized to maximize its satisfaction. Using Schur complement, the inequality \eqref{eq.Schure_5_1} yields the inequality \eqref{eq.closedLoop_condition_DD_MV_1}. Moreover, using $Y_j=G_{K,j}P_j$, $X_0 G_{K,j}=I$ amounts to $X_0 Y_j=P_j$, which gives the equality \eqref{eq.closedLoop_condition_DD_MV_6}. 

Now, consider a matrix $H_j$ such that 
\begin{equation}\label{eq.Schur_final_1}
Y_j P_j^{-1} Y_j^T \preceq H_j.
\end{equation}
which results in
\begin{equation}\label{eq.Schur_final_2}
\operatorname{Tr}\big(Y_j P_j^{-1} Y_j^T\big) \leq \operatorname{Tr}(H_j) \leq \eta_j^2-1. 
\end{equation}

A sufficient condition for the satisfaction of the inequality \eqref{eq.Schur_final_2} is
\begin{equation}\label{eq.Schur_final_3}
\eta_j^2-1 \leq \zeta_j. 
\end{equation}
for which $\zeta_j$ can be minimized to maximize its satisfaction.

Applying the Schur complement on \eqref{eq.Schur_final_1} and \eqref{eq.Schur_final_3} yield the LMIs \eqref{eq.closedLoop_condition_DD_MV_7} and \eqref{eq.closedLoop_condition_DD_MV_8}, respectively, with respect to the constraint \eqref{eq.closedLoop_condition_DD_MV_9}. Based on Lemma 1, Problem 1 is solved. It should be noted that since index $j$ is the circular form of index $i$ and it belongs to the same domain, indices of the decision variables of constraints \eqref{eq.closedLoop_condition_DD_MV_6}--\eqref{eq.closedLoop_condition_DD_MV_9} have been denoted by $i$ for simplicity.
\end{proof}

\begin{remark}
While the acceptable risk level \( \epsilon \) is specified as a fixed parameter during controller synthesis, it does not directly control the true probability of constraint satisfaction in the presence of stochastic disturbances. In practice, the actual risk of safety violations is influenced by the variance of the noise distribution. An increase in noise variance leads to a broader dispersion of the state trajectories, which can elevate the probability of violating safety constraints, even if \( \epsilon \) remains unchanged. This underscores the importance of incorporating the noise covariance structure into the controller design to ensure that the intended probabilistic safety guarantees are reliably achieved.
\end{remark}

\section{Set Partitioning and State-Feedback Gains Calculation}
Up to this point, we have shown the maximization of the convex hull of ellipsoids within the admissible set, rendering it either CSiE or CSiP. Now, we aim to elucidate the process of partitioning and computing state-feedback controllers for these partitions, a procedure that is common for both model-based and data-driven scenarios. For partitioning of the obtained convex hull of ellipsoids, this section generalizes the method given in \cite{hoai2023further}, which is described for second-order systems, to higher-order systems by providing an algorithmic approach. To do so, the following definition is first given.



\begin{definition}
A point $v^*$ belonging to the boundary of $\mathcal{C}$, i.e., $\operatorname{Fr}(\mathcal{C})$, stands as an extreme point of $\mathcal{C}$ if it cannot be expressed as a combination formed through convex combinations of other points within $\mathcal{C}$.
\end{definition}

First step is to find the vertices of the convex hull, which can be achieved by solving the following set of equations for $i=1,\ldots,n_v$  \cite{hoai2023further}
\begin{equation}\label{eq.vertices}
v^T P_i v=1
\end{equation}
where $v \in \mathbb{R}^n$ is the solution of \eqref{eq.vertices}. Not all solutions to the aforementioned equation necessarily represent vertices of the convex hull. Those solutions that do correspond to vertices are denoted by $v^*$. Algorithm 1 summarizes the partitioning method which is performed offline.
\begin{algorithm}
\caption{Set partitioning algorithm}\label{alg.set}
\begin{algorithmic}[1]
\Require $P_i$: A set of matrices defining the equations for the convex hull; $n_v$: Number of vertices.

\Ensure $v^*$: Vertices that form the convex hull.
\Steps
\State \noindent \textbf{for} $i_1=1:n_v-(n-1)$ \newline
\indent \textbf{for} $i_2=i_1+1:n_v-(n-2)$ \newline
\indent \indent $\vdots$ \newline
\indent \indent \textbf{for} $i_n=i_{n-1}+1:n_v-(n-i_n)$ \textbf{do} \newline
\indent \indent \text{$\vartriangleright$ Solve the following set of equations:}
\begin{align}
\phi^T P_{i_1} & \phi^T=1 \nonumber \\
\phi^T P_{i_2} & \phi^T=1 \nonumber \\
& \vdots \nonumber \\
\phi^T P_{i_n} & \phi^T=1 \nonumber
\end{align}
\indent \indent \text{$\vartriangleright$ Obtain all possible vertices of each iteration:}
\begin{align}
v_{i_1} = \; & P_{i_1}\phi \nonumber \\
v_{i_2} = \; & P_{i_2}\phi \nonumber \\
& \vdots \nonumber \\
v_{i_n} = \; & P_{i_n}\phi \nonumber
\end{align}
\indent \indent \text{$\vartriangleright$ Stack all possible vertices of each iteration:}
\begin{equation}
v_{com}=[v_{i_1},v_{i_2},\ldots,v_{i_n}] \nonumber
\end{equation}
\indent \indent \text{$\vartriangleright$ Stack all possible vertices of all iterations:}
\begin{equation}
v_{all}=[v_{all},v_{com}] \nonumber
\end{equation} 
\indent \indent \textbf{end for} \newline
\indent \textbf{end for} \newline
\textbf{end for} \newline
\State \textbf{Use Quickhull algorithm \cite{barber1996quickhull} to find the convex hull of the obtained set of points}
\State \textbf{For each of the extreme points, find $n-1$ neighborhood points. Then, the extreme points $(0,v_1^*,\ldots,v_r^*)$ will be the vertices of the corresponding partition.}
\end{algorithmic}
\end{algorithm}

To design the state-feedback control gains for all partitions, without sacrificing the generality of the situation, let's now examine the scenario where $x(t)$ belongs to the convex combination $\operatorname{Co}(v_1^*,\ldots,v_r^*)$, with $v_i^*$ being extreme points of the convex hull located in $\mathcal{E}(P_i,1)$ for $i=1,\ldots,r$, and where $2 \leq r \leq n$. We can express $x(t)$ as a linear combination
\begin{equation}\label{eq.x_CH}
x(t) = \gamma_1(t) v^*_1 + \ldots + \gamma_r(t) v^*_r 
\end{equation}
where $0 \leq \gamma_i(t) \leq 1$ for $i = 1,\ldots,r$. This representation in equation \eqref{eq.x_CH} can be transformed into a vector form as follows
\begin{equation}\label{eq.x_CH_compact}
x(k) = V^* \Gamma(t)
\end{equation}
where $\Gamma(t) = [\gamma_1(t),\gamma_2(t),\ldots,\gamma_r(t)]^T$, and the matrix $V$ is structured as $V^*=[v^*_1,v^*_2,\ldots,v^*_r]$.

Because $v^*_1,v^*_2,\ldots,v^*_r$ are linearly independent, it is apparent that the rank of matrix $V^*$ is $r$. By utilizing the singular value decomposition (SVD), the matrix $V^* \in \mathbb{R}^{n \times r}$ can be rewritten as
\begin{equation}\label{eq.SVD}
V^* = U^*_v S^*_v {V_v^{*^T}}
\end{equation}
where $U^*_v$ is an $n \times r$ matrix, $V^*_v$ is a $r \times r$ matrix satisfying ${U_v^{*^T}}  U^*_v = I$, ${V_v^{*^T}} V^*_v = I$, and $S^*_v$ is a diagonal matrix with dimensions $r \times r$.

Since the rank of $V^*$ is $r$, it implies that the diagonal elements of $S^*_v$ are all positive. Using equations \eqref{eq.x_CH_compact} and \eqref{eq.SVD}, one can deduce
\begin{equation}\label{eq.Gamma}
\Gamma(t) = V^*_v {S_v^{*^{-1}}} {U_v^{*^T}} x(t)
\end{equation}

On the other hand, the control input for the given $x(t)$ in $n_p$-th partition is computed as
\begin{equation}\label{eq.u_CH}
u_{n_p}(t) = \gamma_1(t) K_1 v^*_1 + \ldots + \gamma_r(t) K_r v^*_r
\end{equation}

Hence, with $u_i = K_i v^*_i$ for all $i = 1,\ldots,r$, and utilizing \eqref{eq.Gamma}, one gets
\begin{equation}\label{eq.u_CH_2}
u_{n_p}(t) = [u_1,u_2,\ldots,u_r] V^*_v {S_v^{*^{-1}}} {U_v^{*^T}} x(t)
\end{equation}
or defining $K_{n_p}^p=[u_1,u_2,\ldots,u_r] V^*_v {S_v^{*^{-1}}} {U_v^{*^T}}$, \eqref{eq.u_CH_2} becomes
\begin{equation}\label{eq.u_CH_3}
u_{n_p}(t) = K_{n_p}^p x(t), \,\,\, \forall n_p=1,\ldots,N_p,
\end{equation}
where $N_p$ shows the number of partitions, and the control gains $K_i$ related to each of the ellipsoids are calculated according to previous theorems. Finally, the piecewise-affine control input is calculated as \eqref{eq.general_controller}.

The achieved results are summarized in Algorithm 2 which is executed online to determine the corresponding state-feedback gain.
 
\begin{algorithm}
\caption{State-feedback gain calculation algorithm}\label{Algorithm_2}
\begin{algorithmic}[1]
\Require Current state $x(t)$; Number of partitions $N_p$; Control gains $K_i$.
\Ensure State-feedback gain $K_{n_p}^p$.
\Steps
\State \noindent \textbf{Examine the current state $x(t)$ to determine the partition to which it belongs.}
\State \textbf{Compute the corresponding control gain as follows}
\begin{align}
K_{n_p}^p=[u_1,u_2,\ldots,u_r] V^*_v {S_v^{*^{-1}}} {U_v^{*^T}}, \,\,\, n_p=1,\ldots,N_p \nonumber
\end{align}
\indent \text{$\vartriangleright$ $N_p$ denotes the number of partitions.}
\end{algorithmic}
\end{algorithm}

\section{Data-based Safe Reinforcement Learning Control Design}
The designed direct data-based risk-aware safe controller is utilized as a safeguard to rectify the actions of readily available RL algorithms with minimal intrusion. In other words, the goal is to limit constraints on the RL agent as much as possible and supervise its actions exclusively in situations where they might potentially jeopardize the safety of the system. To ensure the safety of RL algorithms, we establish an intervention guideline that certifies safety, and importantly, this guideline is independent of the specific RL algorithm chosen.

Additionally, this approach offers robust safety assurances both while training and when deploying RL algorithms. The corrective guideline maintains the RL agent's actions if they are deemed safe. It only interpolates with the data-based secure controller when safety needs to be ensured. This method's advantage lies in the fact that safety validation is only required for the probabilistic controller, allowing the RL agent to concentrate on exploration and learning. The safe controller controller is learned initially (which demands significantly less data compared to training an optimal policy through RL) and remains in use by the intervention guideline throughout the learning process and post the RL agent's learning phase to affirm its safety.

Given that $u^{RL}$ represents the present policy of the RL agent, dictating its actions, the interpolation guideline with minimum intervention generates the control action as follows
\begin{align}\label{eq.u_optimal_safe}
u(t) = \left\{ \begin{array}{l}
u^{RL}(t) \,\,\, \mathrm{if} \,\,\, \mathbb{P}\big[x(t+1) \in \mathcal{S}_c] \geq (1-\epsilon), \\
u^s(t) \,\,\,\,\,\,\,\, \text{otherwise},
\end{array} \right.
\end{align} 
where $\epsilon$ is an acceptable risk level and 
\begin{equation}\label{eq.interpolation}
u^s(t)=\varphi(t) u^{safe}(t)+(1-\varphi(t)) u^{RL}(t)  
\end{equation}
interpolates the safe and optimal controllers using the following linear optimization problem
\begin{align}\label{eq.interpolation}
& \min \,\,\, \varphi(t) \\
& \mathrm{s.t.} \,\,\, \mathbb{P}\big[x(t+1) \in \mathcal{S}_c\big|u^s(t)\big] \geq (1-\epsilon) \nonumber 
\end{align} 

Here, $u^{safe}(t)$ represents the control action executed by the safe controller that has been acquired through data-driven learning, and $\varphi(t)$ is the interpolation variable. The condition $\mathbb{P}\big[x(t+1) \in \mathcal{S}_c\big|u^{RL}(t)\big] \geq (1-\epsilon)$ ensures that the system's state will remain within the safe set, meeting an acceptable threshold, at the subsequent time step $t+1$ after applying the RL action $u^{RL}(t)$ to the system. The scalar interpolation approach is designed not only to ensure safety but also to maintain as much of the original RL policy’s optimal performance as possible. By optimizing a scalar interpolation factor \( \varphi(t) \), the method determines the minimal intervention necessary from the safe controller to satisfy a high-probability safety constraint. When the RL action is already safe, the scalar naturally resolves to zero, ensuring full reliance on the RL policy. Conversely, when safety may be violated, \( \varphi(t) \) increases just enough to restore safety. This principled, low-dimensional interpolation technique makes it possible to operate near the RL controller's performance envelope while avoiding overly conservative behavior typical of traditional safe control schemes. This design choice preserves optimality in expectation, which is particularly important in high-reward or exploration-heavy tasks. A challenge is that knowing the next step requires the knowledge of the $B$ dynamics as discussed next. Learning the $B$ dynamics cannot be achieved under Assumption 4 and requires \eqref{eq.rank} to be full rank. The advantage of the presented approach is that a robust optimization can be performed over an uncertain set of $B$ matrices that are available as prior knowledge, as elaborated in the next assumption. In the following we show how the $B$ dynamics are required and its uncertainties can be incorporated. Note that if more data becomes available, then, a more accurate $B$ can be learned to reduce the conservatism. However, our approach does not need to wait until rich data are collected to make safe decisions. 

\begin{assumption}\label{Assumption_6}
Assume that the input matrix $B$ follows a normal distribution, i.e., $B \sim \mathcal{N}(B_n,\Delta B)$ where $B_n$ is the expected value of $B$, and $\Delta B$ represents its covariance.
\end{assumption}

According to Lemma 2 and due to the fact that 
\begin{equation}
x(t+1)=(A+BK)x(t)+B(u^{RL}-u^{safe})+w(t),
\end{equation}
Similar to the proof of Lemma 3, the random variable $\tilde{\upsilon}_j(t+1)$ in the presence of the RL input is given as
\begin{equation}
\tilde{\upsilon}_j(t+1)=-W_0G_{K,j}\upsilon_j(t)+\Delta B(u^{RL}-u^{safe})+w(t),
\end{equation}
and its covariance, using \eqref{eq.ellips_state_variance_2}, is computed as
\begin{align}
& \mathbb{E}\big[\tilde{\upsilon}_j(t+1)\tilde{\upsilon}_j^T(t+1)\big] \leq \nonumber \\
& \operatorname{Tr}(G_{K,j}P_jG_{K,j}^T)\Sigma+\Sigma+\Delta B(u^{RL}-u^{safe})(u^{RL}-u^{safe})^T \Delta B^T \nonumber \\
& = V_R
\end{align}

Hence, the constraint in \eqref{eq.interpolation} is equivalent to
\begin{equation}\label{eq.safety_criteria}
F_{CH,s}\big(X_1G_{K,p}x(t)+B_n(u^{RL}-u^{safe})\big) \leq (g_{CH,s}-\gamma_s),
\end{equation}
where $F_{CH,s}$ and $g_{CH,s}$ denote the $s$th row of $F_{CH}$ and $g_{CH}$, respectively, and $\gamma_s = \kappa_s\sqrt{F_{CH,s}V_RF_{CH,s}^T}$ with $\kappa_s=\sqrt{\frac{1-\epsilon_s}{\epsilon_s}}$.

The inequality \eqref{eq.safety_criteria} will be used as a safety criteria to check if the next state is likely to violate the safe set by applying the RL policy.

Thus, the control input is computed as
\begin{align}\label{eq.u_optimal_safe_final}
u(t) = \left\{ \begin{array}{l}
u^{RL}(t) \,\,\, \mathrm{if} \,\,\, \Psi_s\big(x(t+1)|u^{RL}(t)\big) \leq (g_{CH,s}-\gamma_s), \\
u^s(t) \,\,\,\,\,\,\,\, \text{otherwise},
\end{array} \right.
\end{align}
with
\begin{align}\label{eq.interpolation_final}
& \min \,\,\, \varphi(t) \\
& \mathrm{s.t.} \,\,\, \Psi_s\big(x(t+1)|u^{s}(t)\big) \leq (g_{CH,s}-\gamma_s). \nonumber
\end{align}
where $\Psi_s\big(x(t+1)|u^{RL}(t)\big)=F_{CH,s}\big(X_1G_{K,p}x(t)+B_n(u^{RL}-u^{safe})\big)$ demonstrates the Minkowski function for the convex hull of ellipsoids when the optimal policy is applied to the system.

\begin{theorem}\label{Theorem_5}
Assume the reinforcement learning agent is designed such that it ensures convergence to the optimal control solution without constraints. Applying the control policy \eqref{eq.u_optimal_safe_final} to the system \eqref{eq.LTI} effectively addresses Problem 1.
\end{theorem}

\begin{proof}
 The condition \eqref{eq.u_optimal_safe_final} provided by the interpolation rule is equivalent to
\begin{align}\label{eq.u_optimal_safe_final_2}
u(t) = \left\{ \begin{array}{l}
u^{RL}(t) \,\,\, \mathrm{if} \,\,\, x(t) \in \{\mathcal{S}_c-\Gamma\}, \\
u^s(t) \,\,\,\,\,\,\,\, \text{otherwise},
\end{array} \right.
\end{align}
where $\Gamma=\{x(t):\Psi_s\big(Ax(t)+Bu^{RL}(t)\big) > (g_{CH,s}-\gamma_s) \}$. Given that ${u^s}$ guarantees safety, the condition ${\Psi_s}\big(Ax(t)+Bu^{RL}(t)\big) \leq (g_{CH,s}-\gamma_s)$ preserve the invariant property of the convex hull.
Furthermore, the safe backup policy only comes into play alongside the RL agent when safety becomes compromised. This empowers an RL agent equipped with guaranteed convergence to explore without constraints, facilitating the acquisition of knowledge for a secure optimal controller. Thus, the solution to Problem 1 is effectively achieved.
\end{proof}

Figure~2 illustrates the overall architecture of the proposed framework, where risk-neutral and risk-aware safe controllers are synthesized from data and integrated with an RL policy via scalar optimization. This structure enables flexible fusion of safety and performance objectives under uncertainty.

\begin{figure}[h]
\centering
\resizebox{\textwidth}{!}{  
\begin{tikzpicture}[node distance=1.4cm and 1.4cm]

\tikzstyle{block} = [rectangle, draw=black, minimum width=3.6cm, minimum height=1.2cm, text centered]
\tikzstyle{arrow} = [->, very thick]

\node (data) [draw=black, fill=orange!30, ellipse, minimum width=4.5cm, minimum height=1.4cm] {Data Collection ($X_0$, $U_0$, $X_1$)};

\node (dyn) [block, fill=blue!15, below=of data, yshift=0cm] {Closed-Loop Dynamics Learning};
\node (ellipsoids) [block, fill=blue!15, below=of dyn, yshift=0cm] {Ellipsoid Construction \& Convex Hull};
\node (partition) [block, fill=blue!15, below=of ellipsoids, yshift=0cm] {Partitioning into Regions};

\node (riskneutral) [block, fill=green!20, below left=of partition, xshift=2.4cm] {Risk-Neutral Control (Certainty Equivalence)};
\node (riskaware) [block, fill=green!20, below right=of partition, xshift=-2.4cm] {Risk-Aware Control (Minimum Variance)};

\node (rl) [block, fill=yellow!30, below=of partition, yshift=-2cm] {RL Policy Learning ($u_{RL}(t)$)};

\node (fusion1) [block, fill=red!15, below left=of rl, xshift=2.4cm] {Controller Merging Using Scalar Optimization};
\node (fusion2) [block, fill=red!15, below right=of rl, xshift=-2.4cm] {Controller Merging Using Scalar Optimization};

\node (output1) [block, fill=gray!20, below=0.6cm of fusion1] {Final Risk-Neutral Safe Optimal Control Output};
\node (output2) [block, fill=gray!20, below=0.6cm of fusion2] {Final Risk-Aware Safe Optimal Control Output};

\draw [arrow] (data) -- (dyn);
\draw [arrow] (dyn) -- (ellipsoids);
\draw [arrow] (ellipsoids) -- (partition);

\draw [arrow] (partition) -- (riskneutral);
\draw [arrow] (partition) -- (riskaware);

\draw [arrow] (riskneutral.south) -- (fusion1.north);
\draw [arrow] (riskaware.south) -- (fusion2.north);

\draw [arrow] (rl.south) -- ++(0,-0.5) -| (fusion1.north);
\draw [arrow] (rl.south) -- ++(0,-0.5) -| (fusion2.north);

\draw [arrow] (fusion1) -- (output1);
\draw [arrow] (fusion2) -- (output2);

\end{tikzpicture}
}  
\caption{Flowchart showing risk-neutral and risk-aware safe control strategies with RL integration.}
\label{fig.flowchart}
\end{figure}
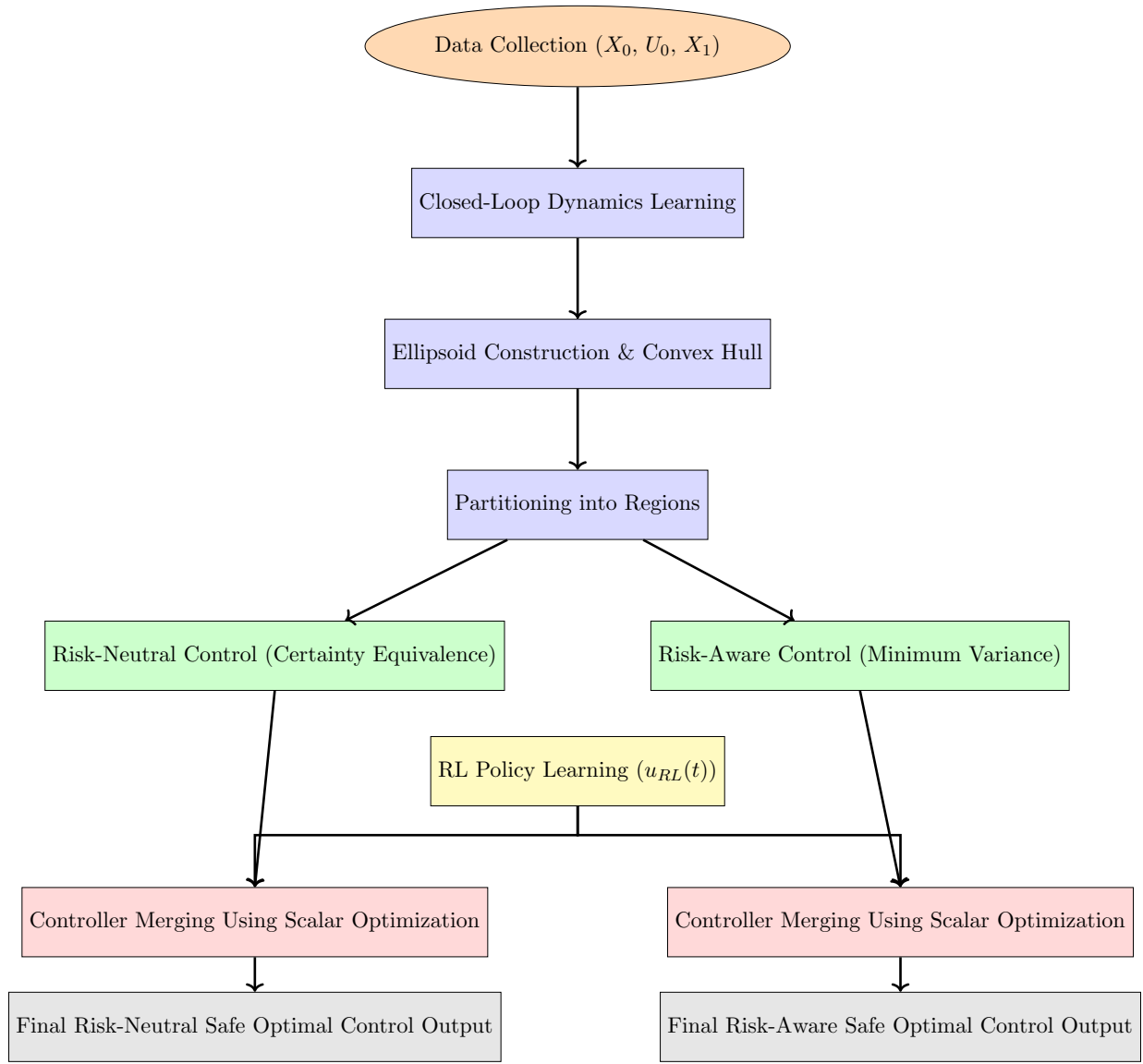

\section{Simulation}
In this section, two simulation examples are provided to evaluate the efficiency of the designed approach in the presence of noise.

\subsection{Numerical example}
Consider the following discrete-time LTI system
\begin{equation}\label{eq.Example}
x(t+1) = \begin{bmatrix}
 0.2895 & -0.0001 \\
 -1.6012 & 0.0295
     \end{bmatrix} x(t) + \begin{bmatrix}
                            0 \\
                            1
                      \end{bmatrix} u(t) + w(t)
\end{equation}

And the admissible set is a polyhedral set defined in \eqref{poly} with
\begin{align}
    F=\begin{bmatrix}
      1/3 & 1/4 \\
      0 & 1/4 \\
     -4/12 & -1/12 \\
     -1/3 & -1/4 \\
      0 & -1/4 \\
      4/12 & 1/12
    \end{bmatrix},
\end{align}
\begin{align}
    g=\begin{bmatrix}
      1 & 1 & 1 & 1 & 1 & 1
    \end{bmatrix}^T.
\end{align}

Using the traditional data-based $\lambda$-contractive approach given in \cite{bisoffi2020data} for this admissible set, the optimization problem becomes infeasible, meaning that there is no state-feedback gain that can stabilize the system \eqref{eq.Example} within this safe set.

Now, to show the efficacy of the convex hull of ellipsoids approach in the open-loop manner, Theorem 1 is applied to the open-loop and deterministic system \eqref{eq.Example}, and the maximum convex hull of ellipsoids is shown in Figure~\ref{fig.openLoop}. The convex hull of ellipsoids obtained in the open-loop form still cannot cover the entire safe set. However, as depicted in Figure~\ref{fig.closedLoop}, by applying the closed-loop form outlined in Theorem 2, using three ellipsoids—or, in other words, three state-feedback control policies—the convex hull of the ellipsoids obtained in closed-loop form almost covers the main polyhedral safe set and becomes $\lambda$-contractive.

To perform the simulation, it is assumed that the noise $w(t)$ follows a Gaussian distribution with a variance of $0.0005I$, where $\lambda=0.8$ and $\delta=0.1$. As the first step of the closed-loop simulation, the performance of the probabilistic safety backup is compared to that of the certainty equivalence method, without taking into account the optimal policy. To maintain fairness in the comparison and highlight the robustness of the minimum-variance method, it is important to emphasize that the safe control approach presented in Theorem 2 is executed without incorporating any measurements of noise.

Figure~\ref{fig.partitioned_set} illustrates the convex hull of ellipsoids and its partitioned form derived using Theorem 4 and Algorithm 1. This partitioning enables the construction of a piecewise-affine safe controller, allowing a broader portion of the admissible set to be covered while accommodating the nonlinearities and stochasticity inherent in the system. Figure~\ref{fig.comparison_1} presents the time evolution of system trajectories over 100 different realizations of Gaussian noise under both the certainty-equivalence safe controller (which disregards variance in its synthesis) and the proposed minimum variance-based probabilistic safe controller. The figure demonstrates that while the certainty-equivalence approach maintains a nominal level of safety in idealized settings, it fails to account for variability, resulting in frequent constraint violations under stochastic disturbances. In contrast, the proposed controller explicitly incorporates noise variance into its synthesis process, thereby minimizing the risk of safety violations and ensuring robustness with high confidence across stochastic realizations.

To further validate the practical utility of the proposed approach, an additional simulation is carried out using the data-based interpolation algorithm described in Theorem 5. In this experiment, we apply the learned unconstrained optimal control policy from \cite{de2021low}—both in isolation and in combination with the proposed safety framework—under Gaussian noise with a covariance of \( \Sigma = 0.01 I \). The cost function weights for the LQR controller are selected as
\begin{align}\label{eq.LQR_params}
Q=\begin{bmatrix}
        100 & 0 \\
        0   & 0.01
\end{bmatrix}, \quad R = 50.
\end{align}

The results, shown in Figure~\ref{fig.comparison_2}, clearly illustrate the limitations of the unconstrained optimal controller, which, in the absence of a safety mechanism, frequently violates state constraints due to its lack of variance-awareness. By integrating this optimal policy with our proposed safety backup controller through a scalar convex combination, the resulting safe optimal policy successfully preserves the performance benefits of the optimal controller while ensuring constraint satisfaction. This integration is achieved through a data-driven scalar optimization framework that minimizes the closed-loop variance, thereby balancing safety and performance in a principled manner.

Furthermore, to isolate and highlight the contribution of the safety controller, Figure~\ref{fig.comparison_2}~(b) also includes the trajectory of the system governed solely by the safety controller (i.e., without RL intervention). This additional trajectory underscores the ability of the safe controller to maintain robust safety guarantees independently, while the merged controller in the safe optimal case further leverages RL-driven optimality with minimal interference. Collectively, these results underscore the effectiveness of the proposed framework in mitigating risk under uncertainty, outperforming traditional methods that either ignore noise variance or impose overly conservative constraints. To quantitatively evaluate performance retention alongside safety, the expected value of the quadratic cost function \( J_s \) is defined as
\begin{equation}\label{eq.cost_LQR}
J_s = \mathbb{E}\left[\sum\limits_{t=0}^\infty x^\top(t) Q x(t) + u^\top(t) R u(t)\right],
\end{equation}
and is computed for different controllers. Since the proposed data-driven safe control framework aims to minimize intervention with the RL agent—unlike traditional CBF-based Safe RL methods that often override actions—\( J_s \) provides a direct measure of how much of the RL policy’s optimality is preserved.

In addition to the cost metric, we also report a safety-compliance score, defined as the number of simulation runs (out of 100 realizations of Gaussian noise) in which the system remained within the admissible set throughout the simulation. As summarized in Table~\ref{table_1}, the purely optimal controller achieves the lowest cost but fails to satisfy safety in all runs, resulting in 0 out of 100 safety-compliant trials. In contrast, the proposed minimum variance-based probabilistic safe optimal controller maintains full safety compliance while incurring only a slight increase in cost, thereby achieving an effective trade-off between safety and performance.

To further demonstrate the superiority of the proposed minimum variance-based probabilistic safe optimal controller, we conduct a comparative analysis with the certainty-equivalence safe control strategy presented in \cite{de2021low}. In this comparison, the certainty-equivalence controller of \cite{de2021low} is first merged with the optimal controller using the same scalar optimization framework described in our method to ensure a fair comparison. Both approaches are evaluated under the same stochastic setup, using Gaussian noise with covariance \( \Sigma = 0.03 I \), an initial state of \( x(0) = [3.30, -1.25]^\top \), and tested across 100 different independent realizations. As illustrated in Figure~\ref{fig.comparison_3}, subfigure (a) shows that the method in \cite{de2021low} fails to ensure safety, resulting in constraint violations due to the absence of variance-aware synthesis. In contrast, subfigure (b) demonstrates that the proposed minimum variance-based probabilistic controller maintains safety while significantly reducing the variance of the closed-loop trajectories. This outcome reflects the core strength of our approach—explicitly accounting for stochastic uncertainty to minimize the probability of constraint violations. Additionally, while the method in \cite{de2021low} constructs only a single ellipsoidal invariant set (depicted as the blue ellipsoid in Figure~\ref{fig.comparison_3}), which is insufficient to fully cover the admissible set, our method employs a convex hull of multiple ellipsoids, providing broader, less conservative, and more robust safe set coverage.

\begin{figure}
    \centering
    \includegraphics[width=0.8\columnwidth]{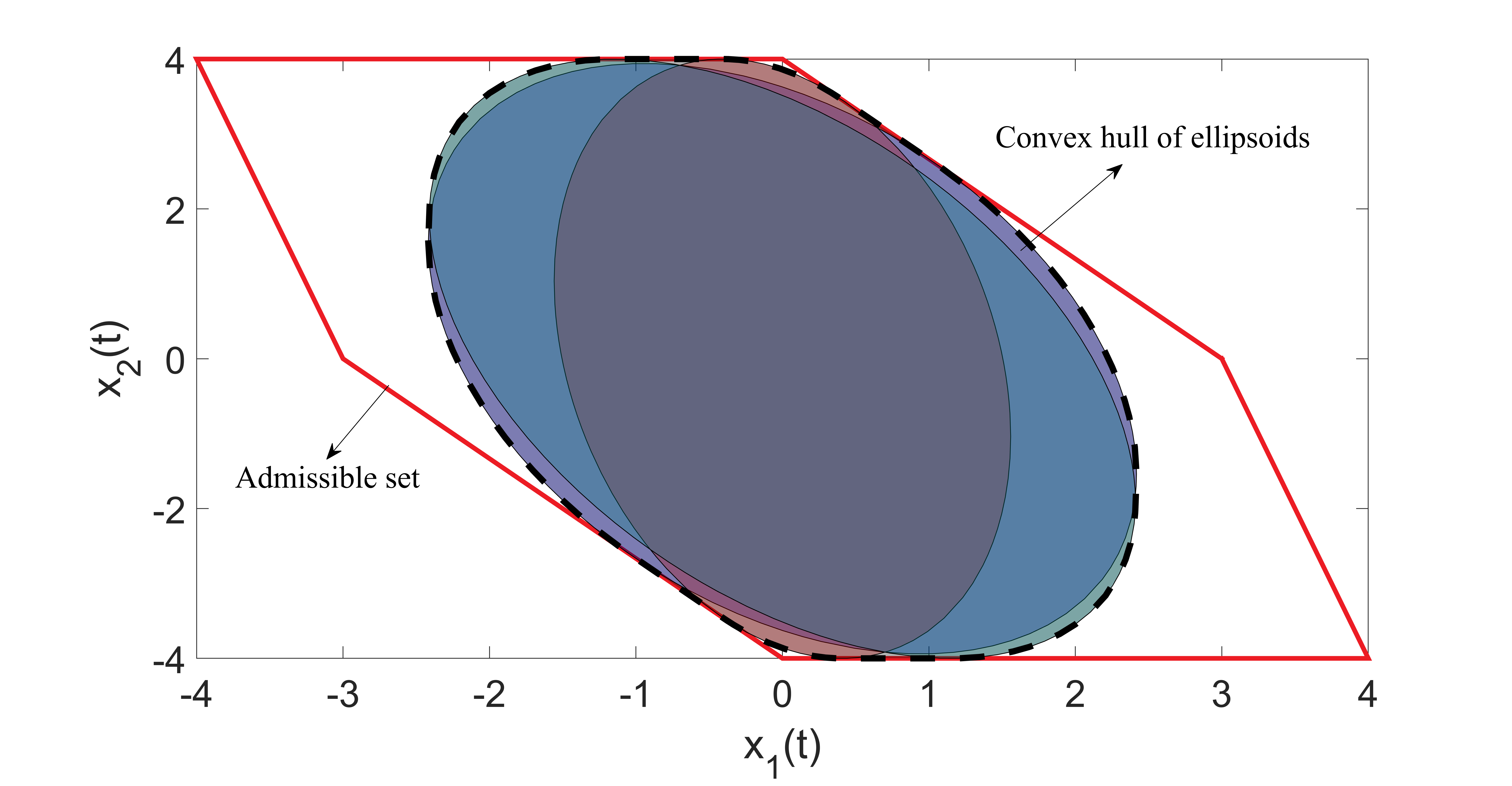}
    \caption{Admissible set containing the ellipsoids and their convex hull obtained using the open-loop method.}
    \label{fig.openLoop}
\end{figure}

\begin{figure}
    \centering
    \includegraphics[width=0.8\columnwidth]{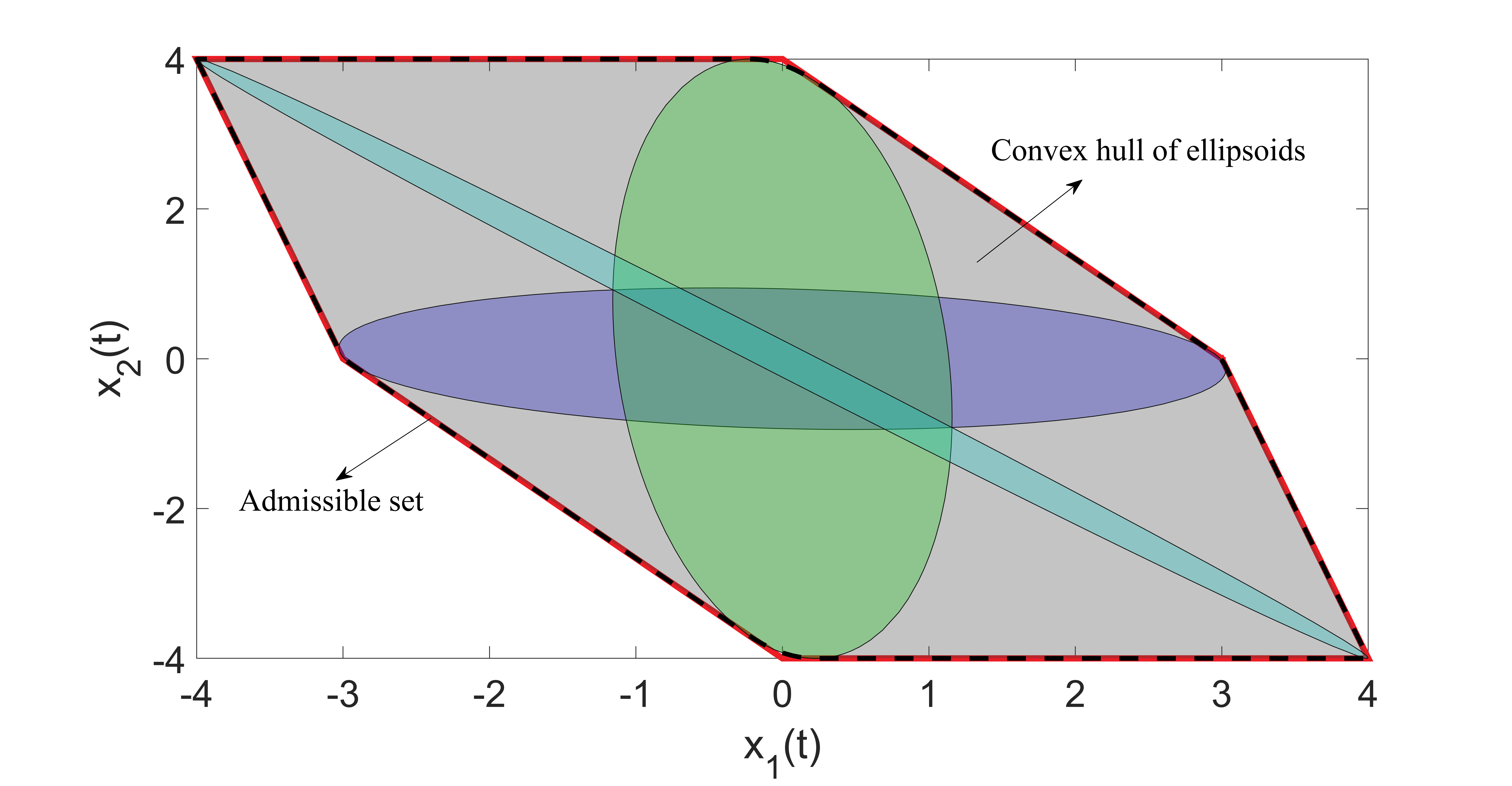}
    \caption{Admissible set containing the ellipsoids and their convex hull obtained using the closed-loop method.}
    \label{fig.closedLoop}
\end{figure}

\begin{figure}
    \centering
    \includegraphics[width=0.8\columnwidth]{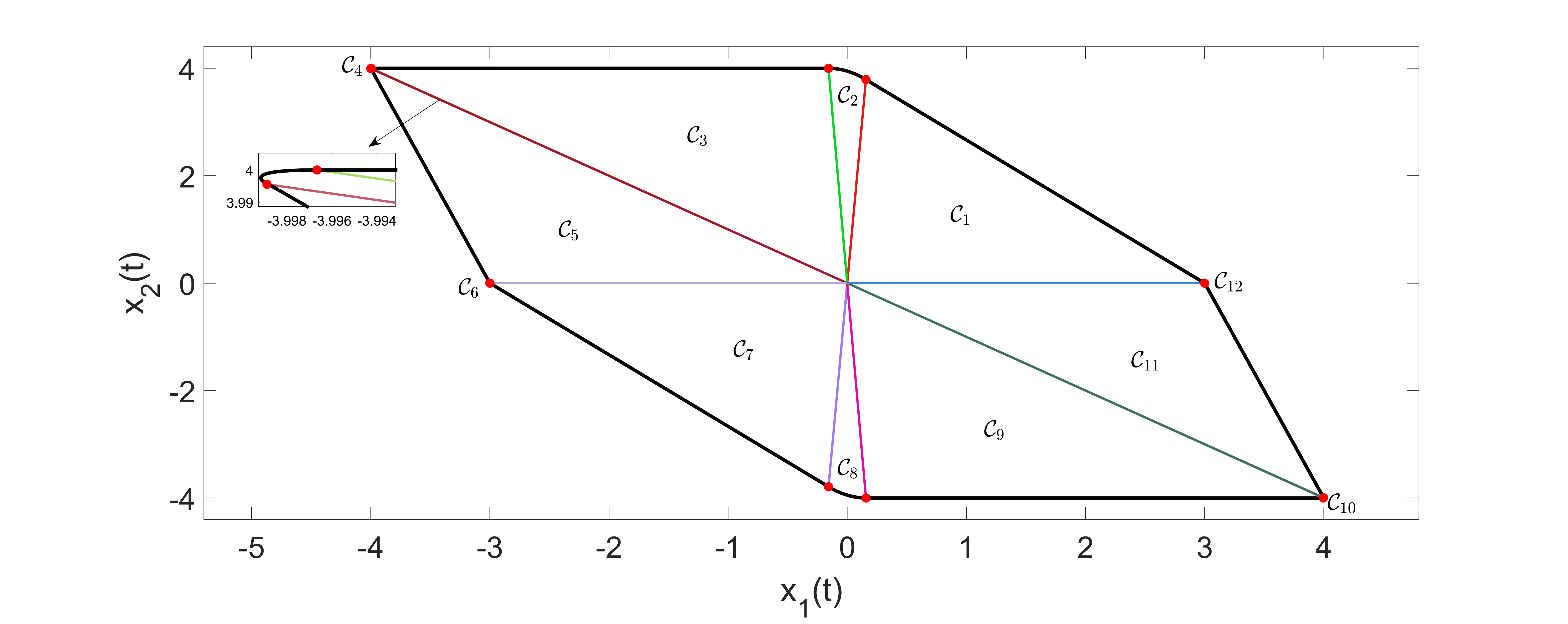}
    \caption{Partitioned convex hull of ellipsoids using Algorithm 1.}
    \label{fig.partitioned_set}
\end{figure}

\begin{figure}
\vspace{-10pt}
         \centering
        {\includegraphics[width=0.8\columnwidth]{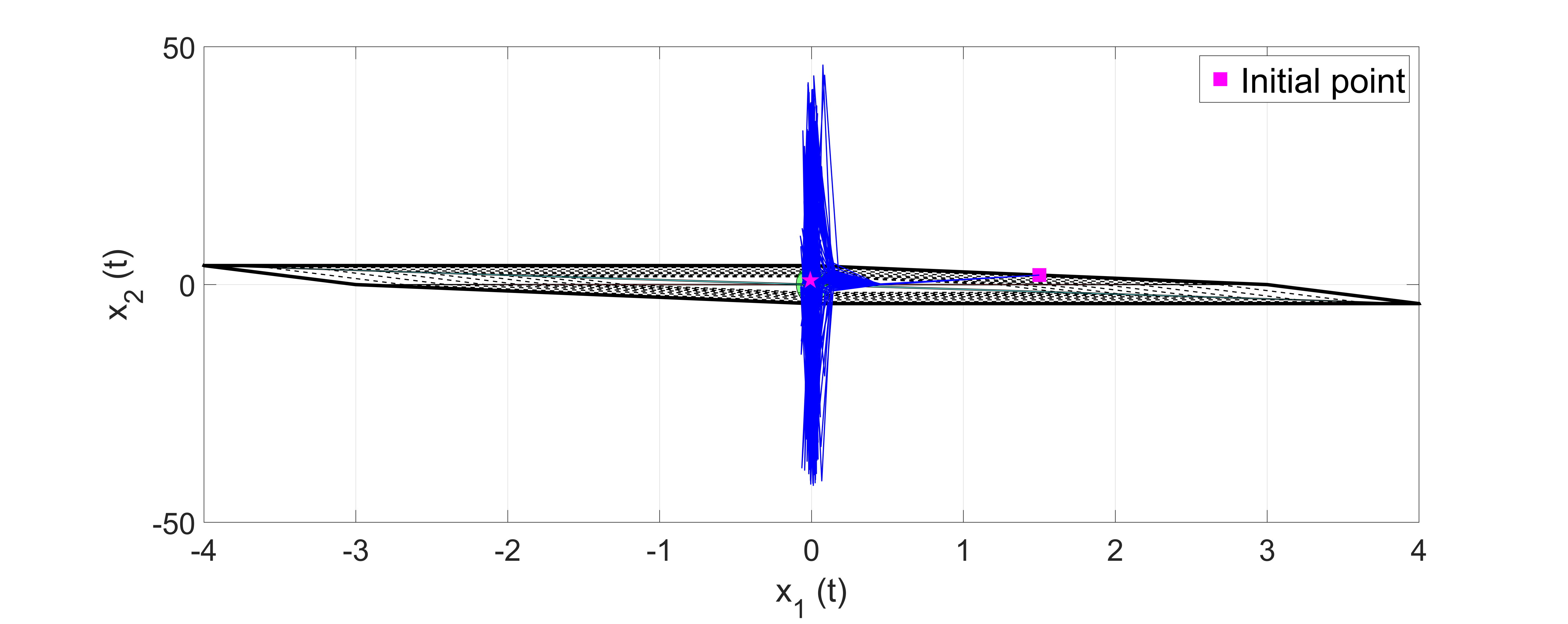}
        \subcaption{}
        \label{suba}}\par  \vspace{-2pt}
        {\includegraphics[width=0.8\columnwidth]{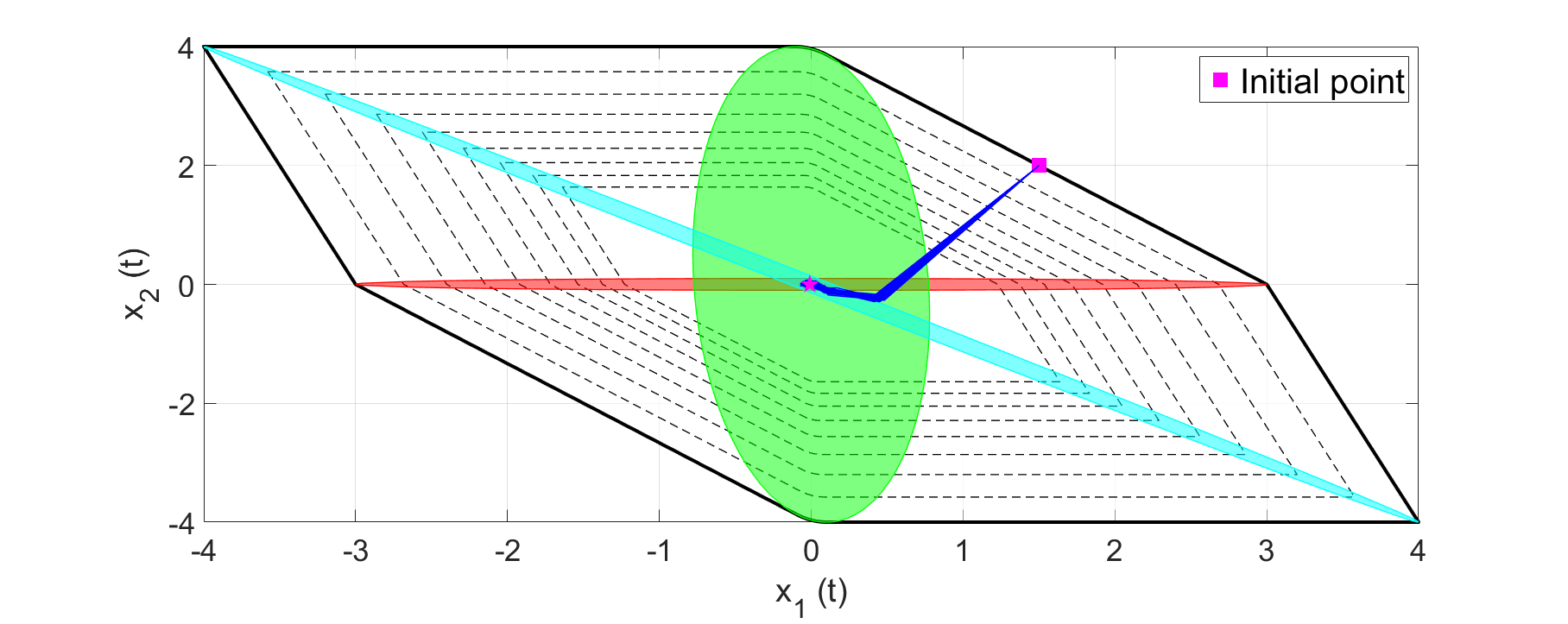}
        \subcaption{}
        \label{subb}} \vspace{-4pt}
                       \caption{Time evolution of the system trajectories under 100 realizations of Gaussian noise with \( \Sigma = 0.0005 I \). Subfigure (a) corresponds to the certainty-equivalence safe controller, which does not account for variance in its synthesis and thus exhibits frequent constraint violations under stochastic disturbances. Subfigure (b) shows the performance of the proposed minimum variance-based probabilistic safe controller, which explicitly incorporates noise variance to ensure robust constraint satisfaction and significantly reduce the risk of safety violations.}
        \label{fig.comparison_1}
\end{figure}

\begin{figure}
\vspace{-10pt}
         \centering
        {\includegraphics[width=0.8\columnwidth]{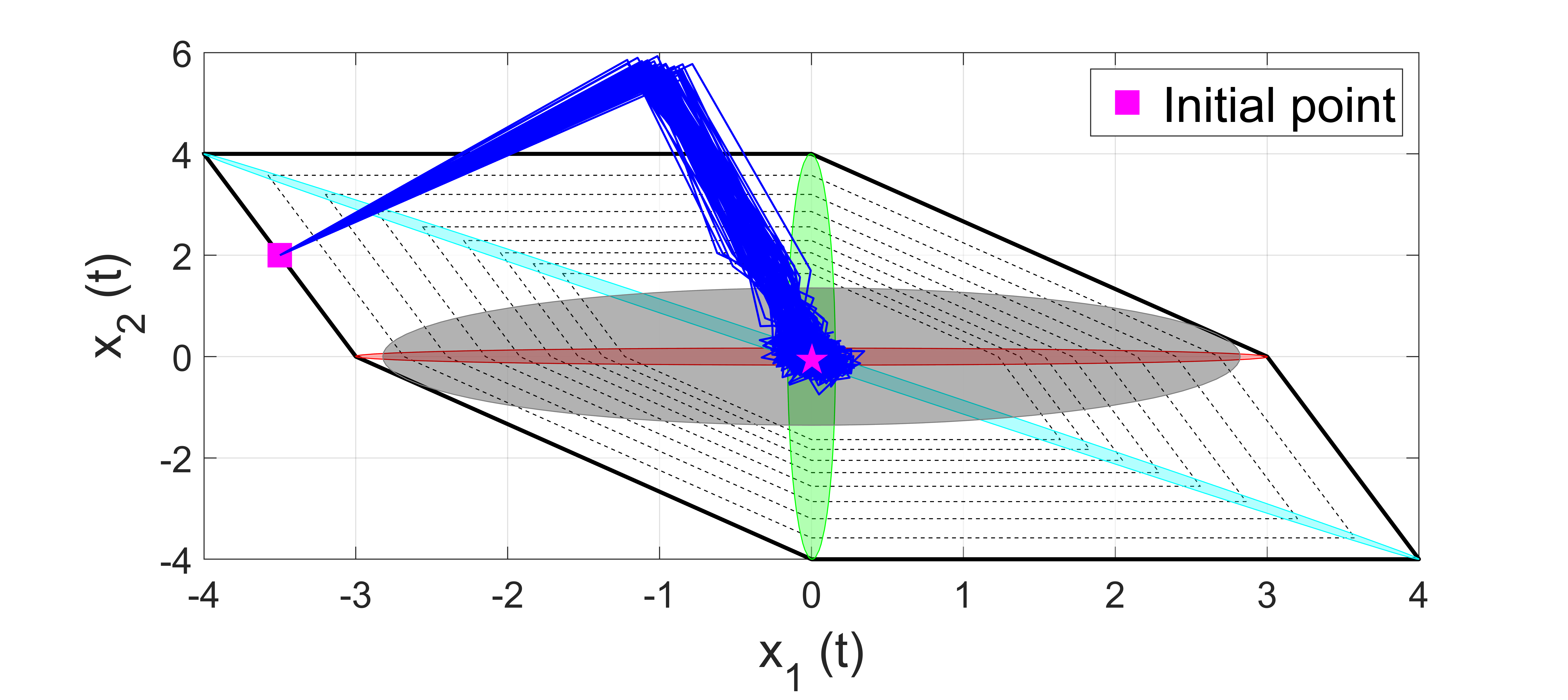}
        \subcaption{}
        \label{suba}}\par  \vspace{-2pt}
        {\includegraphics[width=0.8\columnwidth]{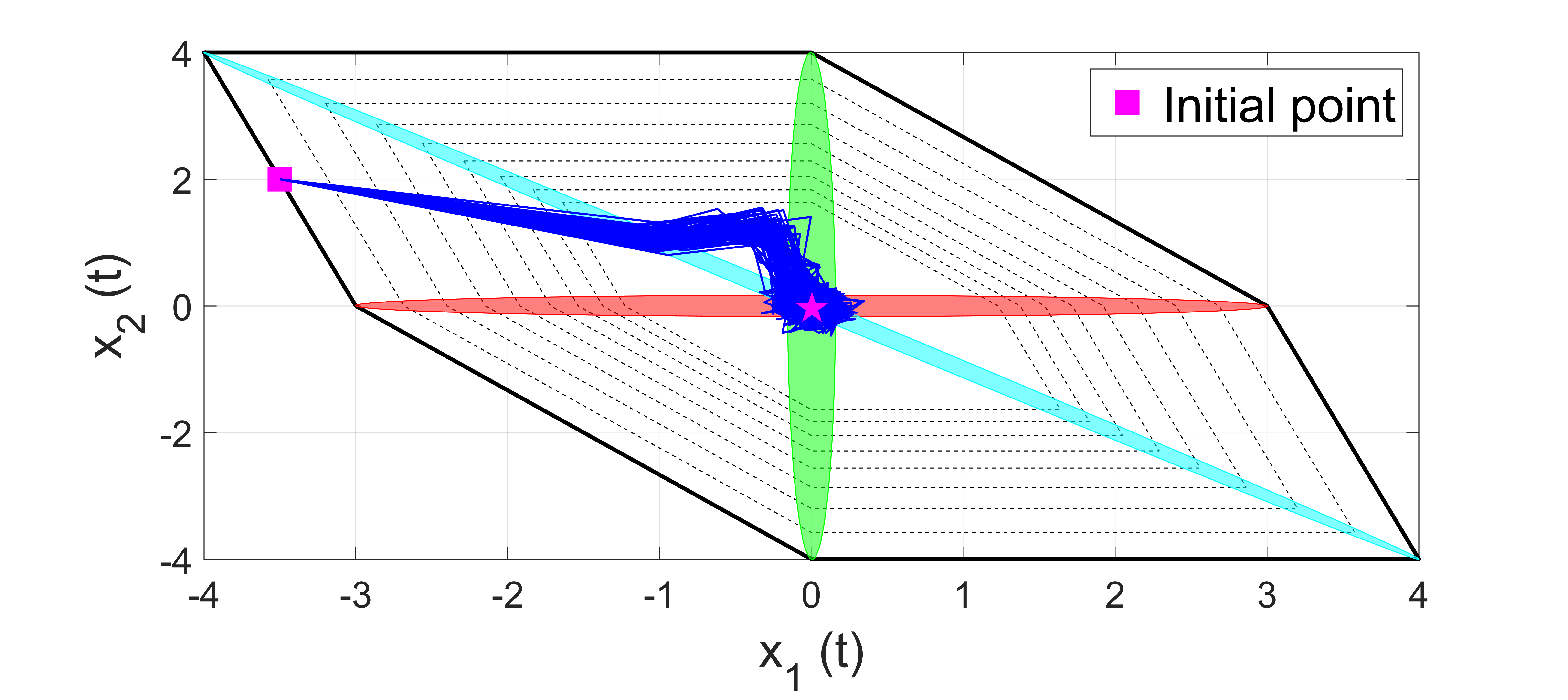}
        \subcaption{}
        \label{suba}}\par  \vspace{-2pt}
        {\includegraphics[width=0.8\columnwidth]{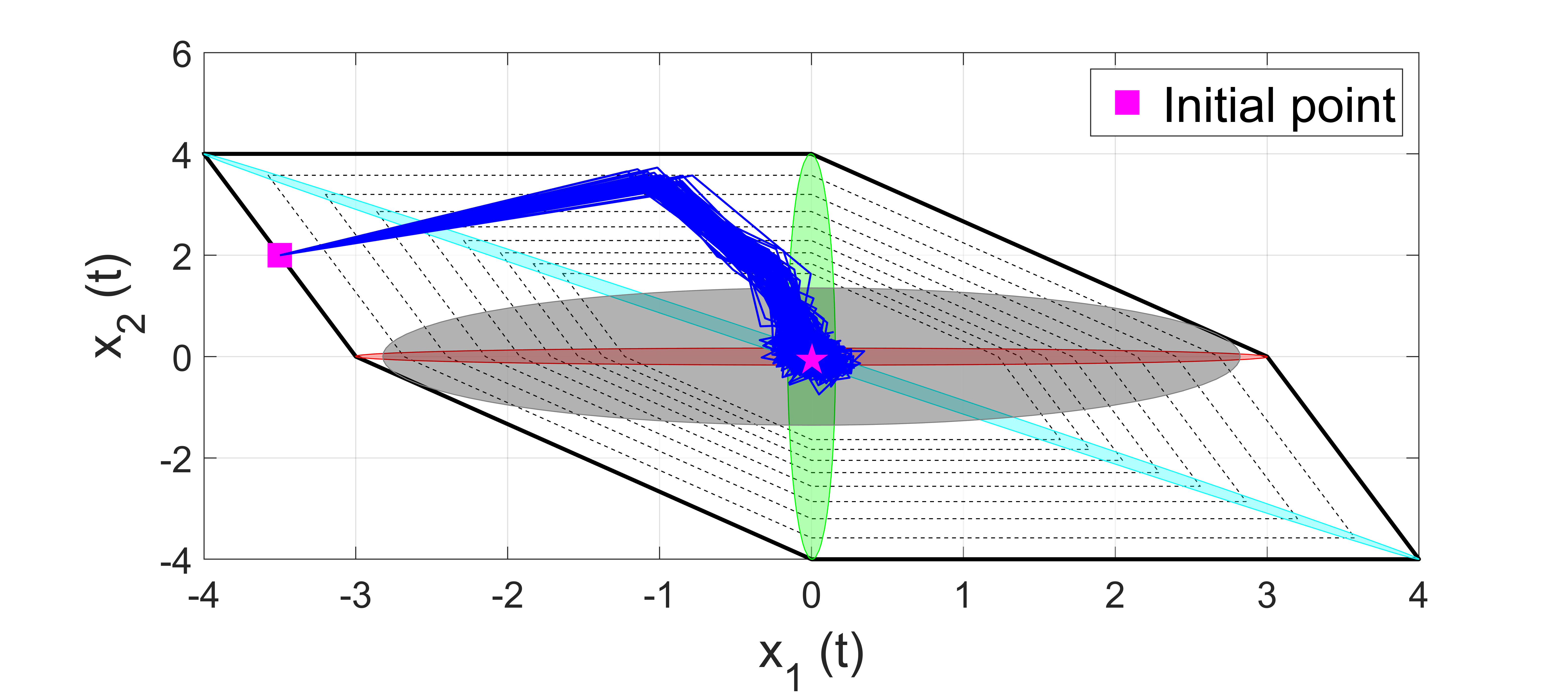}
        \subcaption{}
        \label{subb}} \vspace{-4pt}
                       \caption{Time history of the system states for 100 realizations of Gaussian noise with \( \Sigma = 0.01 I \), illustrating the performance of three controllers: (a) the unconstrained optimal controller, which frequently violates constraints due to the absence of variance-awareness; (b) the minimum variance-based probabilistic safe controller, which ensures constraint satisfaction by minimizing safety violation variance; and (c) the proposed minimum variance-based probabilistic safe optimal controller, which integrates the optimal policy with the safety controller using a data-driven scalar optimization. This integration balances performance and safety by preserving the benefits of the optimal controller while robustly satisfying safety constraints. The gray ellipsoid represents the largest optimal invariant set, and the remaining ellipsoids depict those forming the convex hull.}

        \label{fig.comparison_2}
\end{figure}

\begin{figure}
\vspace{-10pt}
         \centering
        {\includegraphics[width=0.8\columnwidth]{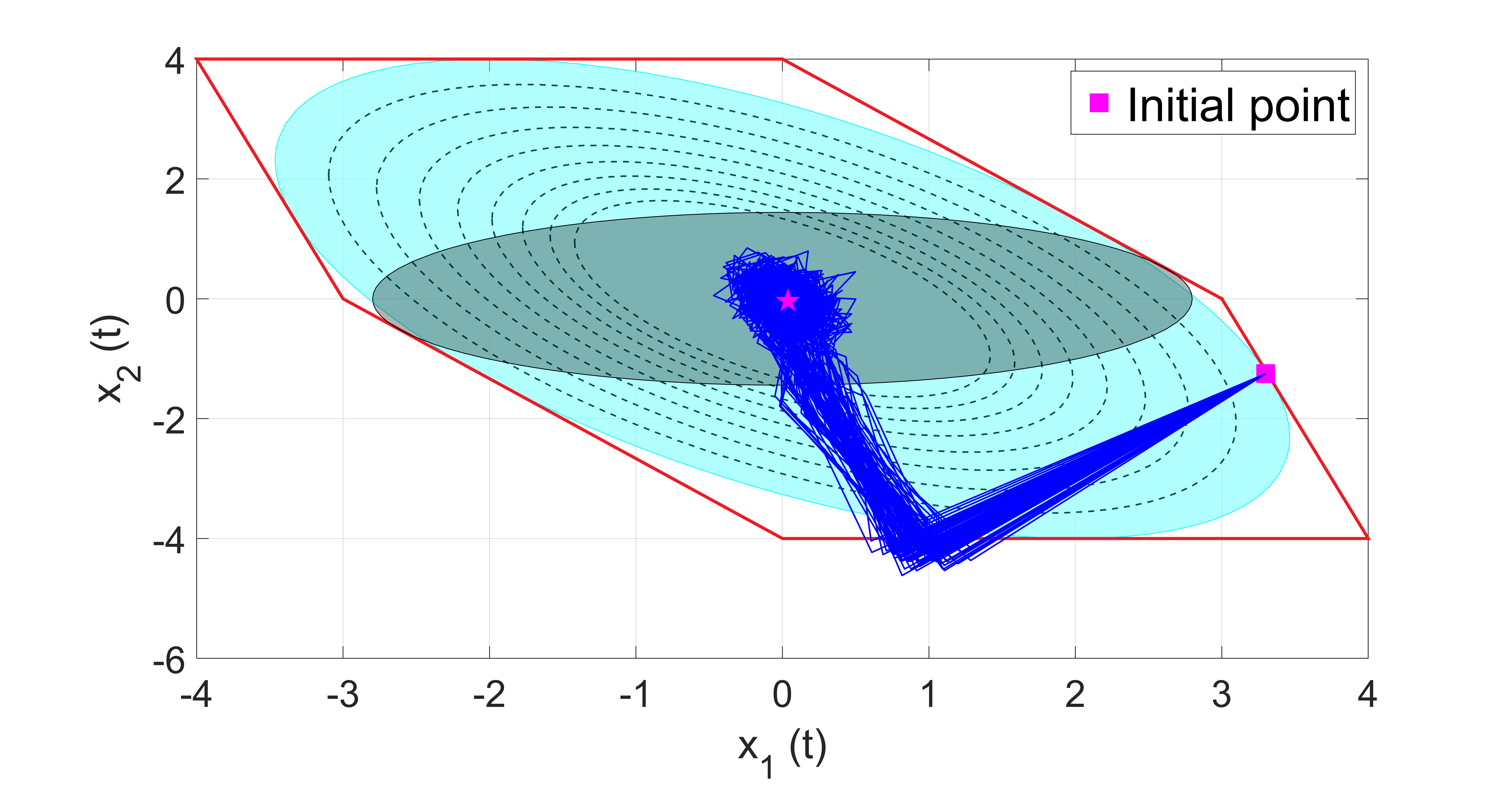}
        \subcaption{}
        \label{suba}}\par  \vspace{-2pt}
        {\includegraphics[width=0.8\columnwidth]{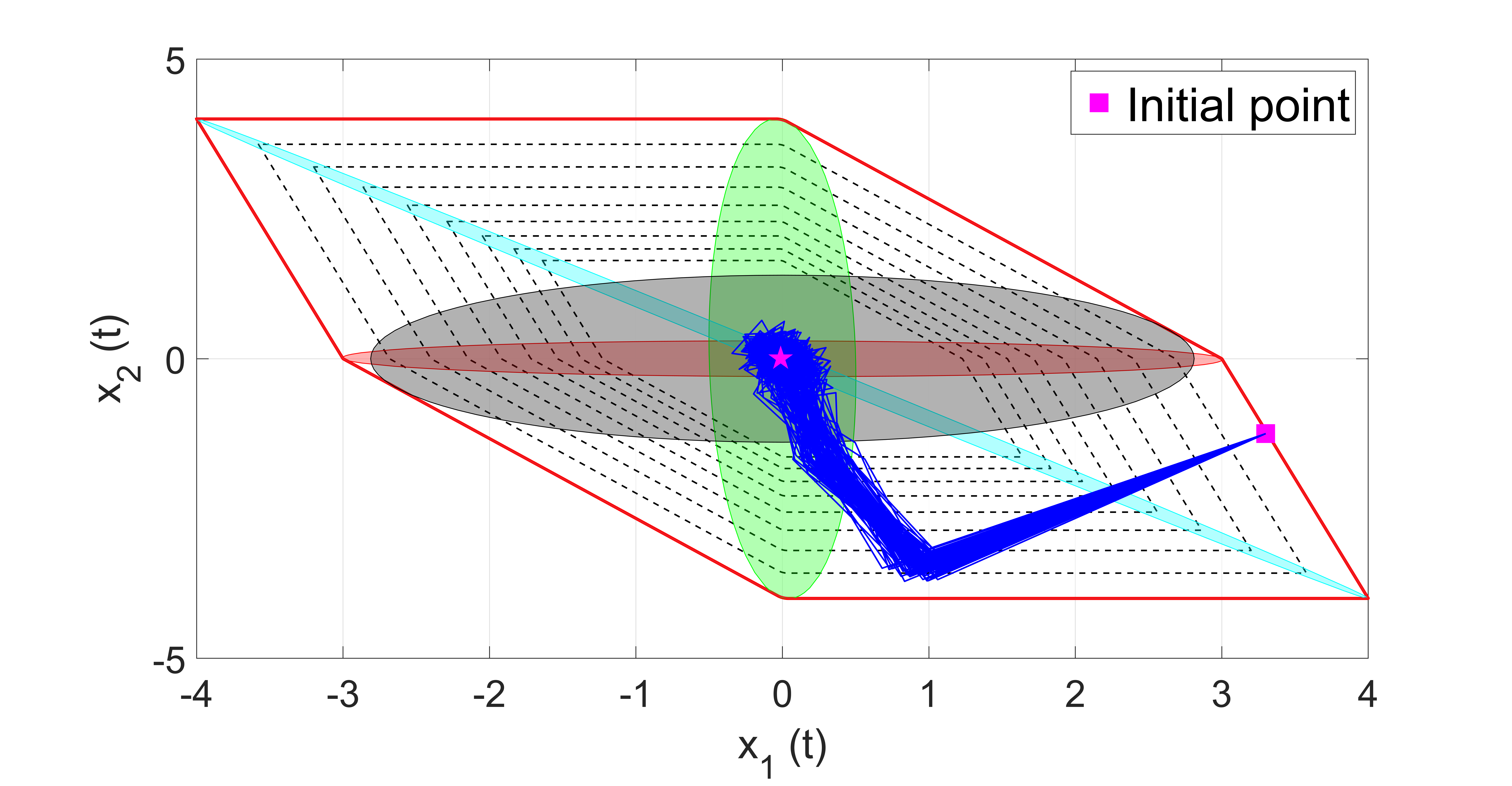}
        \subcaption{}
        \label{suba}}\par  \vspace{-2pt}
                       \caption{Comparison between (a) the certainty-equivalence safe control method of \cite{de2021low} and (b) the proposed minimum variance-based probabilistic safe controller, under 100 different realizations of Gaussian noise with \( \Sigma = 0.03 I \). The proposed method maintains safety by constructing a convex hull of multiple ellipsoids that collectively approximate the admissible set and reduce variance. In contrast, the method in \cite{de2021low} generates only a single ellipsoid (shown in blue), which fails to fully cover the admissible set and results in safety violations under stochastic disturbances. This comparison highlights the improved robustness and safety of the proposed approach.}
        \label{fig.comparison_3}
\end{figure}



\begin{table}[h]
\caption{Quantitative Comparison of Controllers Based on Performance and Safety Compliance}
\label{table_1}
\centering
\begin{tabular}{|c|c|c|}
    \hline
    \textbf{Controller} & \textbf{Expected Cost \( J_s \)} & \textbf{Safety-Compliant Trials (out of 100)} \\
    \hline
    Optimal Controller      & 136270  & 0   \\
    Safe Controller         & 137360  & 100 \\
    Safe Optimal Controller & 136520  & 100 \\
    \hline
\end{tabular}
\end{table}
\vspace{-6pt}

\subsection{Practical Example: Car Lane Keeping Problem}
The lateral dynamics of an autonomous vehicle for a lane-keeping task are modeled by the following discrete-time system \cite{ames2016control}

\[
\begin{aligned}
\begin{bmatrix}
y(t+1) \\
v(t+1) \\
\phi(t+1) \\
\psi(t+1)
\end{bmatrix}
&=
\begin{bmatrix}
1 & T_s & V_0 T_s & 0 \\
0 & 1 + \left( \frac{-C_f + C_r}{M V_0} \right) T_s & 0 & \left( \frac{b C_r - a C_f}{M V_0} - V_0 \right) T_s \\
0 & 0 & 1 & T_s \\
0 & \left( \frac{b C_r - a C_f}{I_z V_0} \right) T_s & 0 & 1
\end{bmatrix}
\begin{bmatrix}
y(t) \\
v(t) \\
\phi(t) \\
\psi(t)
\end{bmatrix}
+
\begin{bmatrix}
0 \\
\frac{C_f}{M} \\
0 \\ 
\frac{a C_f}{I_z}
\end{bmatrix}
T_s u(t)
+
w(t),
\end{aligned}
\]
where \( y(t) \) denotes the lateral displacement, \( v(t) \) is the lateral velocity, \( \phi(t) \) is the yaw angle, and \( \psi(t) \) is the yaw rate at time step \( t \). The control input \( u(t) \) represents the steering angle, and \( w(t) \in \mathbb{R}^4 \) is an exogenous noise accounting for the road curvature. The model parameters are given by \( V_0 = 27.7 \, \mathrm{m/s} \) (longitudinal velocity), \( C_f = 133000 \, \mathrm{N/rad} \) (front cornering stiffness), \( C_r = 98800 \, \mathrm{N/rad} \) (rear cornering stiffness), \( M = 1650 \, \mathrm{kg} \) (vehicle mass), \( I_z = 2315.3 \, \mathrm{kg \cdot m^2} \) (yaw moment of inertia), \( a = 1.11 \, \mathrm{m} \) (distance from the center of gravity to the front axle), \( b = 1.59 \, \mathrm{m} \) (distance from the center of gravity to the rear axle), and \( T_s \) is the sampling time. This model captures the essential lateral and yaw dynamics of the vehicle required for designing a lane-keeping controller under road curvature disturbances.

The objective of the control problem is to maintain the vehicle’s position close to the centerline of the lane while accounting for lateral dynamics and disturbances. By defining the state vector as \( x(t) = [x_1(t), x_2(t), x_3(t), x_4(t)]^\top = [y(t), v(t), \phi(t), \psi(t)]^\top \), the admissible set is specified by safety constraints \( -1.5 < x_1 < 1.5 \) for the lateral displacement and \( -8 < x_2 < 8 \) for the lateral velocity. The controller parameters are also set to $\lambda = 0.84$ and $\delta = 0.1$, and the sampling time is $T_s=0.01 \, \mathrm{s}$.

Unlike the previous 2D example, where the ellipsoidal safe sets and their convex hulls could be directly visualized in the state space, the lane-keeping system considered here is four-dimensional, involving lateral displacement, lateral velocity, yaw angle, and yaw rate. Due to the high-dimensional nature of the system, visualizing the ellipsoids and safe sets in the full state space is not feasible. Therefore, all geometric comparisons and visual representations were restricted to the 2D example, while this 4D example serves as a practical and realistic scenario to evaluate the effectiveness of the proposed method in a high-dimensional setting.

\begin{figure}
\vspace{-10pt}
         \centering
        {\includegraphics[width=0.8\columnwidth]{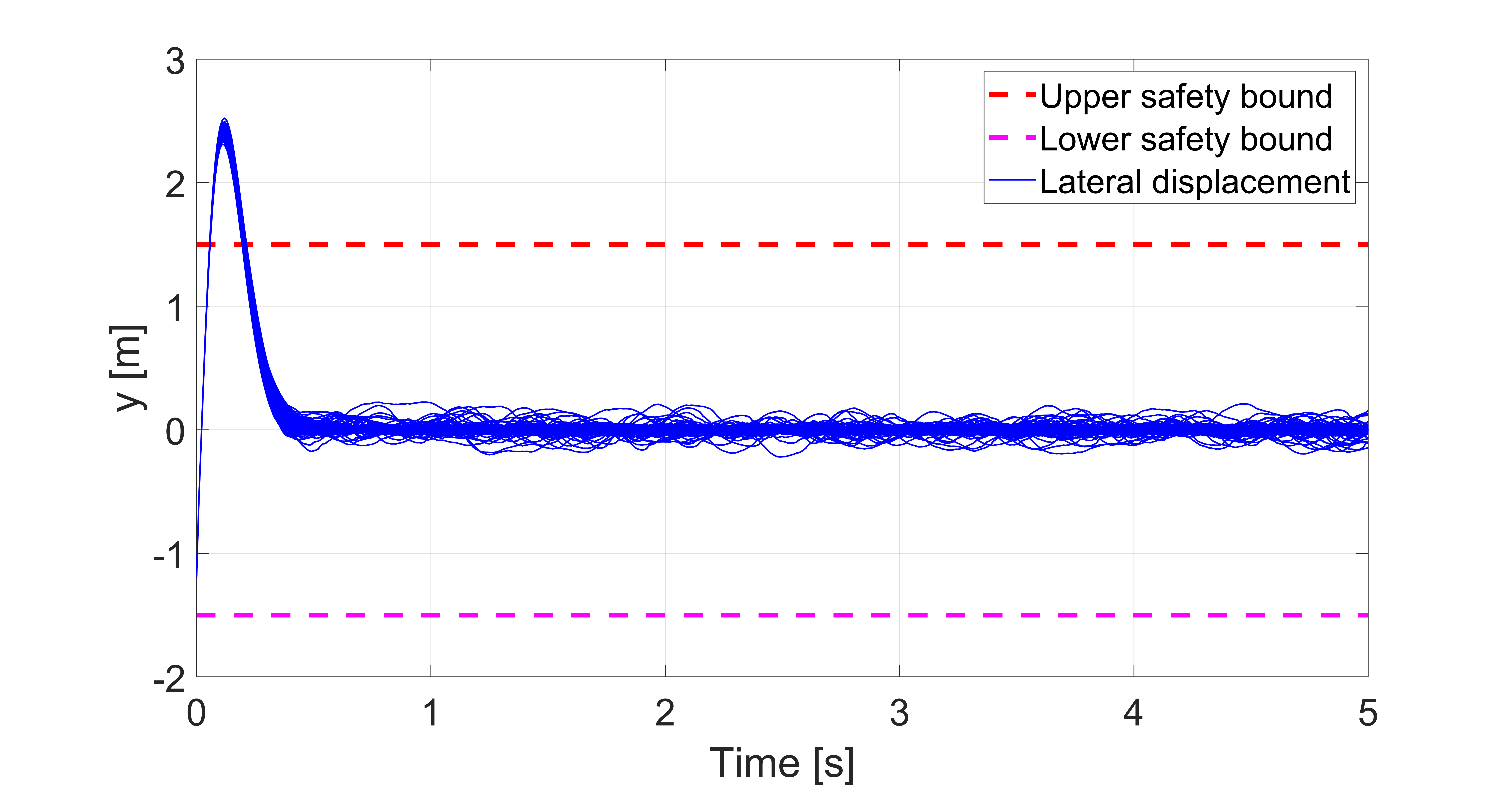}
        \subcaption{}
        \label{suba}}\par  \vspace{-2pt}
        {\includegraphics[width=0.8\columnwidth]{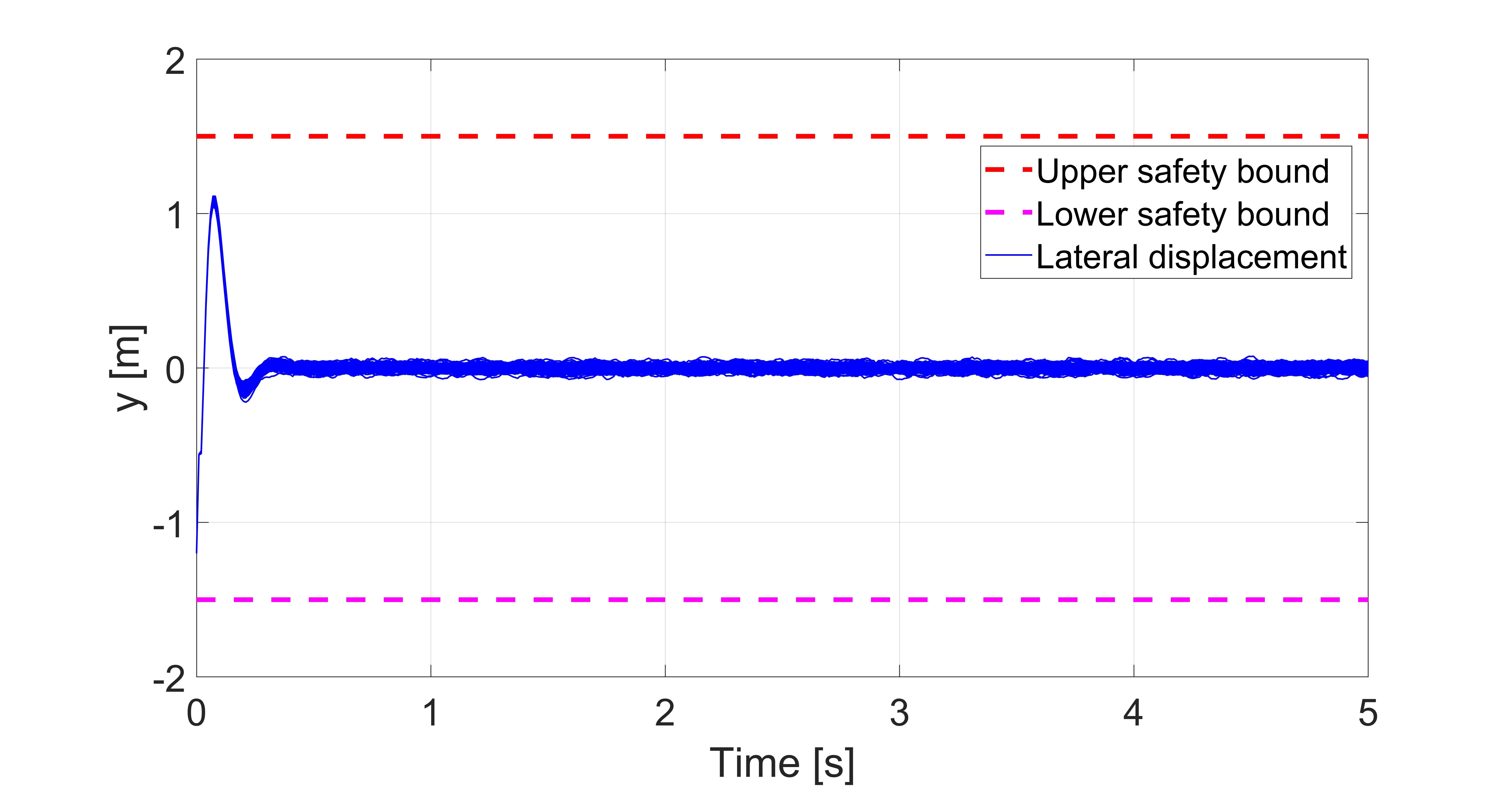}
        \subcaption{}
        \label{suba}}\par  \vspace{-2pt}
                       \caption{Lateral displacement of the vehicle under different control strategies over 100 realizations of Gaussian noise with \( \Sigma = 0.0005I \). Subfigure (a) shows the result of using the purely optimal controller, which violates the lateral safety constraint (\( -1.5 < y < 1.5 \)). Subfigure (b) illustrates the proposed minimum variance-based probabilistic safe optimal controller, which successfully maintains the vehicle's lateral displacement within the admissible safety bounds.}
        \label{fig.lateral_displacement}
\end{figure}

\begin{figure}
\vspace{-10pt}
         \centering
        {\includegraphics[width=0.8\columnwidth]{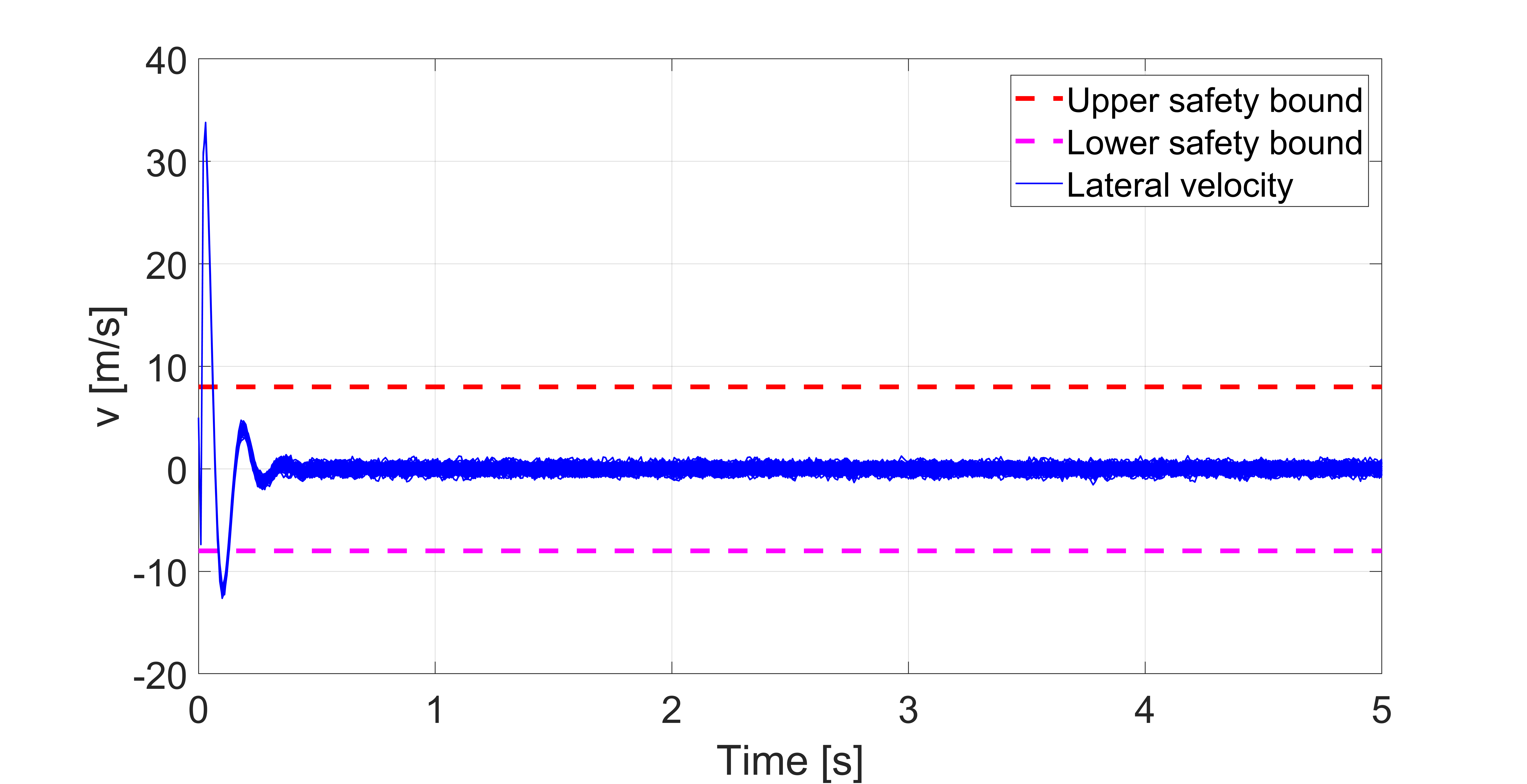}
        \subcaption{}
        \label{suba}}\par  \vspace{-2pt}
        {\includegraphics[width=0.8\columnwidth]{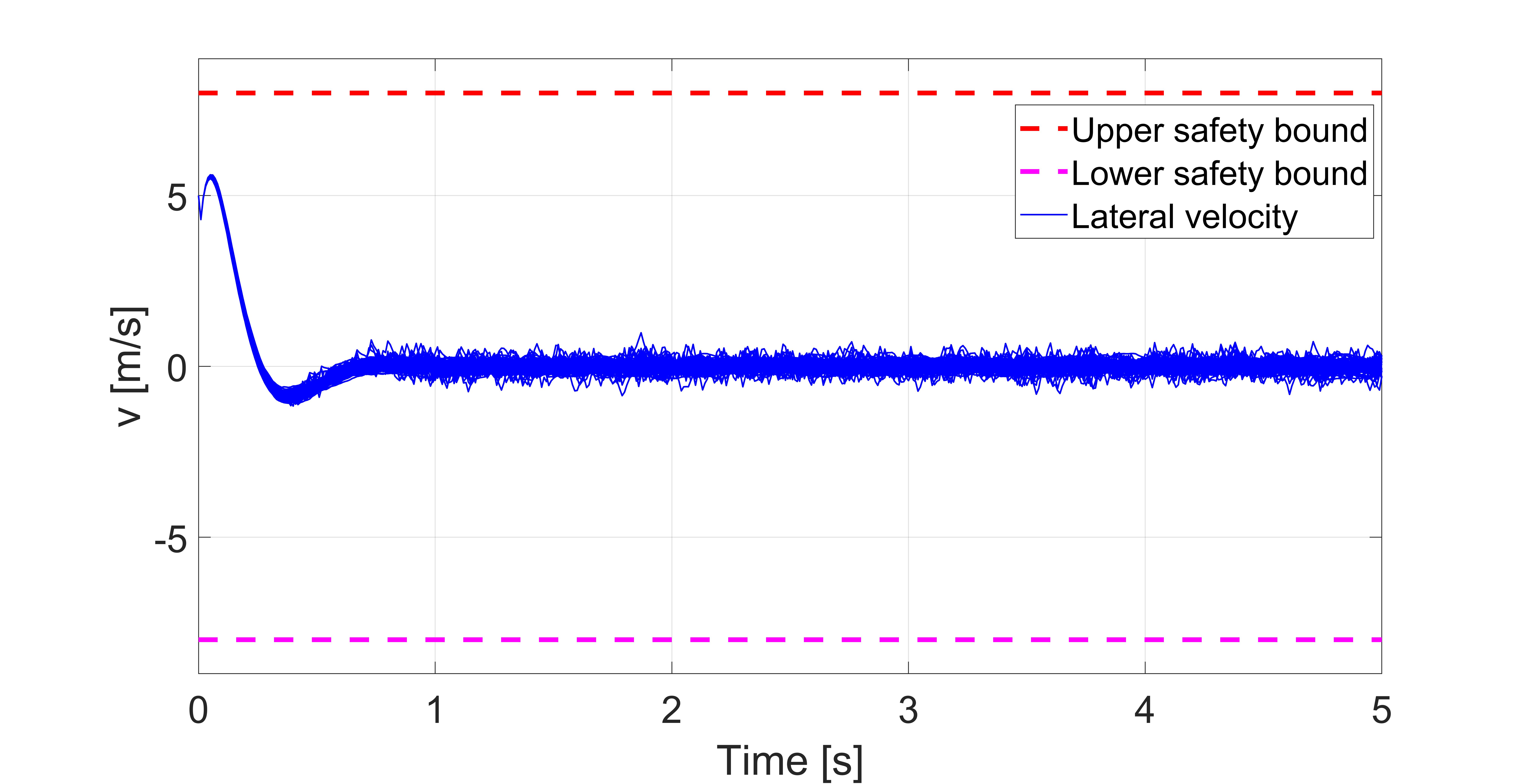}
        \subcaption{}
        \label{suba}}\par  \vspace{-2pt}
                       \caption{Lateral velocity of the vehicle under different control strategies over 100 realizations of Gaussian noise with \( \Sigma = 0.0005I \). Subfigure (a) shows the response under the purely optimal controller, which results in unsafe high lateral velocities. Subfigure (b) shows the performance of the proposed minimum variance-based probabilistic safe optimal controller, which successfully limits the lateral velocity within safe operational bounds.}
        \label{fig.lateral_velocity}
\end{figure}

Figure~\ref{fig.lateral_displacement} illustrates the evolution of the vehicle's lateral displacement under different control strategies over 100 different realizations of Gaussian noise with \( \Sigma = 0.0005I \). As shown in subfigure (a), the purely optimal controller—designed without considering safety—causes the vehicle to drift beyond the safety bounds (e.g., \( -1.5 < y_k < 1.5 \)). In contrast, subfigure (b) demonstrates that the proposed minimum variance-based probabilistic safe optimal controller successfully keeps the lateral displacement within the admissible limits. Figure~\ref{fig.lateral_velocity} presents the corresponding lateral velocity profiles. Subfigure (a) reveals that the purely optimal controller produces unsafe high lateral velocities, whereas subfigure (b) confirms that the safe optimal controller effectively regulates the velocity within a safe range. These results underscore the efficacy of the proposed data-driven safety framework in enforcing probabilistic safety guarantees while maintaining system performance under realistic noisy conditions.

\section{conclusion} \vspace{-2pt}
This paper presents a risk-aware safe reinforcement learning control strategy for stochastic discrete-time linear time-invariant systems. Using the convex hull of ellipsoids, a large portion of the complex admissible sets becomes $\lambda$-contractive in probability, leading to a model-free risk-informed safety backup for RL agents without requiring system model identification. By emphasizing risk-averse control design, minimizing state variance within the closed-loop system, and introducing a data-driven interpolation technique, this approach offers a more robust and efficient solution compared to traditional methods. Unlike conventional myopic safe RL approaches, the proposed framework minimizes intervention with the RL agent to preserve optimal action behavior. Simulation results validate its effectiveness, promising improved safety and performance for reinforcement learning-based control systems in practical, noisy environments. 

Future work will focus on extending the proposed control scheme to accommodate asymmetric admissible sets around the origin, multi-agent systems with coupled constraints, and general nonlinear stochastic dynamics. The latter may involve the use of local linearization techniques or Koopman operator-based modeling to preserve risk-aware safety guarantees in complex environments.

\section*{Acknowledgment}
This work is supported in part by the Department of Navy award N00014-22-1-2159 and in part by the National Science Foundation under award ECCS-2227311.

\bibliographystyle{ieeetr}
\bibliography{Refs}

\end{document}